\documentclass[12pt]{article}

\usepackage{graphicx}
\usepackage[caption=false,font=footnotesize]{subfig}

\usepackage{geometry}
\usepackage{color}
\usepackage{amsmath}
\usepackage{amsfonts}
\usepackage{bbm}
\usepackage[nodisplayskipstretch]{setspace}

\usepackage[scaled]{helvet}
\usepackage[T1]{fontenc}

\usepackage{IEEEtrantools}
\usepackage{algorithm}
\usepackage{algorithmic}

\usepackage{ifthen}
\usepackage{etoolbox}

\usepackage{natbib}

\usepackage{pgfplots}
 \pgfplotsset{compat=newest}
 \pgfplotsset{plot coordinates/math parser=false}
 \usepgfplotslibrary{external}
 \pgfplotsset{try min ticks=2}
 
\graphicspath{{figures/}}

\newtheorem{theorem}{Theorem}[section]

\newtheorem{proposition}[theorem]{Proposition}

\newtheorem{model}[theorem]{Model}

\newenvironment{proof}[1][Proof]{\begin{trivlist}
\item[\hskip \labelsep {\bfseries #1}]}{\end{trivlist}}

\newcommand{\qed}{\nobreak \ifvmode \relax \else
      \ifdim\lastskip<1.5em \hskip-\lastskip
      \hskip1.5em plus0em minus0.5em \fi \nobreak
      \vrule height0.75em width0.5em depth0.25em\fi}

\newcommand{\cent}{\:;\:}

\newcommand{\real}{\mathbb{R}}

\newcommand{\lhood}{l}
\newcommand{\nconst}[1]{K_{#1}}

\newcommand{\lsspace}{\mathcal{X}}
\newcommand{\lsdim}{{d_{\ls{}}}}

\newcommand{\obspace}{\mathcal{Y}}
\newcommand{\obdim}{{d_{\ob{}}}}

\newcommand{\priorden}{p}
\newcommand{\postden}{\pi}
\newcommand{\seqden}[1]{\pi_{#1}}
\newcommand{\seqdenapprox}[1]{\hat{\pi}_{#1}}
\newcommand{\impden}{q}

\newcommand{\lhoodapprox}[1]{\hat{l}_{#1}}

\newcommand{\logprior}{M}
\newcommand{\loglhood}{L}
\newcommand{\logseqden}[1]{\Xi_{#1}}
\newcommand{\logseqdenapprox}[1]{\hat{\Xi}_{#1}}
\newcommand{\logpartden}[1]{\Upsilon_{#1}}
\newcommand{\logpw}[1]{W_{#1}}
\newcommand{\bmden}{\varsigma}

\newcommand{\lsmnapprox}[1]{\hat{m}_{#1}}
\newcommand{\lsvrapprox}[1]{\hat{P}_{#1}}

\newcommand{\flowbm}[1]{\epsilon_{#1}}          
\newcommand{\flowdrift}[1]{\zeta_{#1}}          
\newcommand{\flowdiffuse}[1]{\eta_{#1}}         
\newcommand{\flowcov}[1]{D_{#1}}                
\newcommand{\dsf}{\gamma}                
\newcommand{\flowdriftscale}[1]{A_{#1}}
\newcommand{\flowdriftconst}[1]{b_{#1}}

\newcommand{\lslin}{x^*}
\newcommand{\lgmomapprox}[1]{\hat{\lgmom}_{#1}}       
\newcommand{\obapprox}[1]{\hat{y}_{#1}}

\newcommand{\flowdriftapprox}[1]{\hat{\zeta}_{#1}}      
\newcommand{\flowdiffuseapprox}[1]{\hat{\eta}_{#1}}     

\newcommand{\half}{\frac{1}{2}}

\newcommand{\determ}[1]{\left|#1\right|}

\newcommand{\pd}[3]{\left.\frac{\partial #1}{\partial #2}\right|_{#3}}

\newcommand{\expect}[1]{\mathbb{E}_{#1}}

\newcommand{\bigo}[1]{\mathcal{O}\left(#1\right)}

\DeclareMathOperator{\trace}{Tr}

\newcommand{\rightasconverge}{\stackrel{a.s.}{\rightarrow}}

\newcommand{\ti}{n}

\newcommand{\ls}[1]{x_{#1}}
\newcommand{\ob}[1]{y_{#1}}

\newcommand{\normalden}[3]{\mathcal{N}\left(#1\left|\vphantom{#1}#2, \:#3\right.\right)}

\newcommand{\pss}[1]{^{(#1)}}

\newcommand{\den}{p}

\newcommand{\transden}{f}
\newcommand{\obsden}{g}

\newcommand{\partden}[1]{\upsilon_{#1}}
\newcommand{\pw}[1]{w_{#1}}
\newcommand{\npw}[1]{\bar{w}_{#1}}

\newcommand{\anc}[2]{a_{#1}^{(#2)}}

\newcommand{\numpart}{\mathbf{N}}

\newcommand{\obsfun}{\psi}

\newcommand{\lgmom}{H}
\newcommand{\lgmov}{R}

\newcommand{\pdv}[2]{\frac{\partial #1}{\partial #2}}
\newcommand{\ppdv}[2]{\frac{\partial^2 #1}{\partial #2^2}}
\newcommand{\mpdv}[3]{\frac{\partial^2 #1}{\partial #2 \partial #3}}
\newcommand{\npdv}[3]{\frac{\partial^{#1} #2}{\partial #3^{#1}}}

\newcommand{\pt}{\lambda}                       
\newcommand{\dpt}{\delta\lambda}                
\newcommand{\dls}{\delta x}                     
\newcommand{\dbm}[1]{\delta \epsilon_{#1}}
\newcommand{\pos}[1]{p_{#1}}           
\newcommand{\vel}[1]{v_{#1}}           
\newcommand{\bng}[1]{b_{#1}}               
\newcommand{\rng}[1]{r_{#1}}                    
\newcommand{\hei}[1]{h_{#1}}                    
\newcommand{\rngrt}[1]{s_{#1}}                  
\newcommand{\terrain}{T}                        
\newcommand{\noise}[1]{e_{#1}}

\newcommand{\xS}[1]{r_{S,#1}}
\newcommand{\xE}[1]{r_{E,#1}}
\newcommand{\xH}[1]{r_{H,#1}}
\newcommand{\aB}{\alpha_B}
\newcommand{\aS}{\alpha_S}
\newcommand{\aE}{\alpha_E}
\newcommand{\dU}{d_U}
\newcommand{\dL}{d_L}

\newcommand{\stdnorm}[1]{z_{#1}}

\newcommand{\fixed}{^*}

\newcommand{\lsmn}[1]{m_{#1}}
\newcommand{\lsvr}[1]{P_{#1}}

\newcommand{\expectloglhood}{\expect{\seqden{\pt}}\left[\loglhood\right]}

\newcommand{\errstat}[1]{e_{#1}}

\title{Approximations of the Optimal Importance Density using Gaussian Particle Flow Importance Sampling}
\author{Pete Bunch and Simon Godsill}
\date{}

\begin{document}

\maketitle

\begin{abstract}
Recently developed \emph{particle flow} algorithms provide an alternative to importance sampling for drawing particles from a posterior distribution, and a number of particle filters based on this principle have been proposed. Samples are drawn from the prior and then moved according to some dynamics over an interval of pseudo-time such that their final values are distributed according to the desired posterior. In practice, implementing a particle flow sampler requires multiple layers of approximation, with the result that the final samples do not in general have the correct posterior distribution. In this paper we consider using an \emph{approximate Gaussian flow} for sampling with a class of nonlinear Gaussian models. We use the particle flow within an importance sampler, correcting for the discrepancy between the target and actual densities with importance weights. We present a suitable numerical integration procedure for use with this flow and an accompanying step-size control algorithm. In a filtering context, we use the particle flow to sample from the optimal importance density, rather than the filtering density itself, avoiding the need to make analytical or numerical approximations of the predictive density. Simulations using particle flow importance sampling within a particle filter demonstrate significant improvement over standard approximations of the optimal importance density, and the algorithm falls within the standard sequential Monte Carlo framework.
\end{abstract}



\section{Introduction}

The particle filter is a Monte Carlo algorithm used for sequential inference of a filtering distribution associated with a state-space model. A set of weighted samples is advanced through time, drawn approximately from the filtering distribution. For a comprehensive introduction, see for example \citep{Cappe2007,Doucet2009}. The desired posterior filtering densities contain an intractable normalising constant, which is circumvented through the use of importance sampling. The principal challenge then, when designing a particle filter, is the selection of the importance density.

For nonlinear models, good choices of importance densities are frequently not obvious, particularly when informative observations of the latent state are made. In this situation, simple strategies such as sampling from the prior lead to a set of particles which are spread widely over the state space, of which a large proportion will have very low likelihood. The result is that the variance of the particle weights is high, and the resulting Monte Carlo estimates are dominated by a few particles with high weights. This phenomenon is known as \emph{weight degeneracy}. Although the optimal importance density (OID) which minimises the incremental weight variance is known, it rarely has an analytical form. In practice, Gaussian approximations of the OID based on linearisation or the unscented transform are popular choices for the importance density \citep{Doucet2000a,Merwe2000}, but these are not always effective.

One way in which weight degeneracy may be mitigated is by introducing the effect of each observation gradually, so that particles may be progressively drawn towards peaks in the likelihood. This can be achieved by using a discrete set of \emph{bridging distributions} which transition smoothly between the prior and posterior. Each one is targeted in turn using importance sampling, and the accumulation of weight variance is curtailed through the use of resampling and Markov chain Monte Carlo (MCMC) steps. Such schemes have been suggested by \citet{Neal2001,DelMoral2006} for static inference and by \citet{Godsill2001b,Gall2007,Deutscher2000,Oudjane2000} for particle filters.

It is possible to take the idea of bridging distributions to a limit and define a continuous sequence of distributions between the prior and the posterior. This idea was used by \citet{Gelman1998} for the related task of simulating normalising constants, and has been used to design sophisticated assumed-density filters \citep{Hanebeck2003a,Hanebeck2012,Hagmar2011}. More recently, particle filters have appeared which exploit the same principle, including the \emph{particle flow} methods described in series of papers including \citep{Daum2008,Daum2011d}, and the \emph{optimal transport} methods of \cite{Reich2011,Reich2012a,Reich2013}. A particle is first sampled from the prior (i.e. the transition) density, and then moved continuously according to some differential equation over an interval of \emph{pseudo-time}, such that the evolution in the density corresponds to the progressive introduction of the likelihood.

Although theoretically elegant and powerful, practical implementation of optimal transport or particle flow methods require a host of approximations to be made. First, the expressions for the optimal flow dynamics are the solution to a partial differential equation and are rarely analytically tractable. Second, when applying particle flow to sample from the filtering density, the prior is generally not known analytically, and must itself be approximated. Third, once an appropriate flow has been identified, it must usually then be integrated numerically.

In this paper we focus on models which have a Gaussian prior and likelihood, but a nonlinear relationship between observations and latent states. We move the particles according to an \emph{approximate Gaussian flow}, based on a simple linearisation around each particle state. Unlike existing particle flow algorithms, we do not treat these directly as samples from the posterior, but as proposals in an importance sampler. Thus we obtain an accompanying differential equation for the importance weights in order to correct for the discrepancies introduced by approximating the flow. (We note that \cite{Reich2013} has also recently suggested using a particle flow for importance sampling, but using completely different mechanisms to move the particles and update the weights.) The approximate Gaussian flow cannot be integrated analytically, so we introduce an efficient numerical scheme based on the analytical solution to the linear Gaussian flow, equipped with an effective step size control mechanism. Finally, we apply this particle flow proposal method to the OID of a particle filter, rather than to the filtering density itself. This allows the particle flow to be applied within the standard framework for particle filtering, and also avoids the need to use approximations of the predictive density.

We demonstrate the efficacy of Gaussian flow importance sampling for particle filtering with simulations on a number of challenging nonlinear models. Significant performance improvements are observed in error and effective sample size statistics.

In section~\ref{sec:particle_flow_importance_sampling}, we review importance sampling and particle flow methods. The main exposition on using Gaussian flows for importance sampling is contained in section~\ref{sec:gaussian_flows}. In section~\ref{sec:gaussian_flows_for_particle_filters}, this strategy is applied to particle filtering, and in section~\ref{sec:simulations}, performance is evaluated in a number of challenging simulation studies.

A brief description of a special case of our method has been previously reported in the conference proceedings of CAMSAP \citep{Bunch2013a}.

\section{Importance Sampling and Particle Flows} \label{sec:particle_flow_importance_sampling}

Consider the task of sampling from a Bayesian posterior distribution over a hidden state variable $\ls{} \in \lsspace = \real^\lsdim$,
\begin{IEEEeqnarray}{rClCrCl}
 \postden(\ls{}) & = & \frac{ \priorden(\ls{}) \lhood(\ls{}) }{ \nconst{} } & \qquad\qquad & \nconst{} & = & \int_{\lsspace} \priorden(\ls{}) \lhood(\ls{}) d\ls{} \label{eq:bayesian_posterior}      .
\end{IEEEeqnarray}
in which $\priorden$ and $\postden$ are the prior and posterior densities respectively, which are assumed to exist, $\lhood$ is the likelihood and $\nconst{}$ is a normalising constant, which typically cannot be computed.

\subsection{Importance Sampling}

Importance sampling may be used to draw from posterior distributions \eqref{eq:bayesian_posterior} \citep{Geweke1989,Liu2001a}. A set of $\numpart$ i.i.d. samples $\{\ls{}\pss{i}\}$ (or \emph{particles}, the two terms are used interchangeably throughout) is generated according to some importance distribution with density $\impden(\ls{})$ (whose support is a superset of that of $\postden(\ls{})$) and each is assigned a weight,
\begin{IEEEeqnarray}{rClCrCl}
 \pw{}\pss{i}  & = & \frac{ \priorden(\ls{}\pss{i}) \lhood(\ls{}\pss{i}) }{ \impden(\ls{}\pss{i}) } & \qquad\qquad & \npw{}\pss{i} & = & \frac{ \pw{}\pss{i} }{ \sum_j \pw{}\pss{j} }     .
\end{IEEEeqnarray}
An estimator of a posterior expectation may then be written as a finite sum over this set of weighted samples, and it is well known that this estimate is consistent, converging almost surely to its true value as the number of particles becomes large \citep{Liu2001a},
\begin{IEEEeqnarray}{rCl}
 \sum_{i=1}^{\numpart} \npw{\ti}\pss{i} \phi(\ls{}\pss{i}) & \rightasconverge & \int \postden(\ls{}) \phi(\ls{}) d\ls{}     \label{eq:consistent_estimator}       .
\end{IEEEeqnarray}

The effectiveness of such an importance sampler depends on the choice of importance density. For integration of an arbitrary test function $\phi(\ls{})$, it is desirable that $\impden(\ls{})$ be as close to $\postden(\ls{})$ as possible. Selecting a good importance density is therefore a foremost priority, but often proves challenging. One naive approach is to use the prior as the importance density $\impden(\ls{}) = \priorden(\ls{})$, meaning that $\pw{}\pss{i} = \lhood(\ls{}\pss{i})$. (In a sequential setting, this is the \emph{bootstrap filter} of \cite{Gordon1993}.) This scheme is simple and easy to implement. The only requirement is that it should be possible to sample from the prior. However, it is wasteful, especially when the variance of the prior is much greater than that of the posterior, i.e. the likelihood is highly informative about the state. In this situation, the samples are widely spread over the state space, and only a few fall in the region of high likelihood. The consequence is that many have very low weight and posterior estimates are based on only a few significant particles; the resulting estimators are poor, having a high Monte Carlo variance. This is a fundamental difficulty for importance samplers. Good posterior sampling relies on having a good approximation of the posterior to begin with!

\subsection{Particle Flow Sampling}

Particle flow and optimal transport methods are an alternative mechanism for generating posterior samples. They have been applied to Bayesian filtering and data assimilation problems by \cite{Daum2008,Daum2011d,Daum2013,Reich2011,Reich2012a}. The general principle is to begin with samples from the prior, then to move these according to some dynamics over an interval of pseudo-time such that the final values are distributed according to the posterior. One possible way to achieve this is to define the following geometric density sequence over the pseudo-time interval $\pt \in \left[0,1\right]$,
\begin{IEEEeqnarray}{rClCrCl}
 \seqden{\pt}(\ls{}) & = & \frac{ \priorden(\ls{}) \lhood(\ls{})^{\pt} }{ \nconst{\pt} } & \qquad\qquad & \nconst{\pt} & = & \int_{\lsspace} \priorden(\ls{}) \lhood(\ls{})^{\pt} d\ls{} \label{eq:density_sequence}      .
\end{IEEEeqnarray}
Since $\seqden{0} = \priorden$, initial particles may be sampled from the prior. These are then moved according to an It\={o} stochastic differential equation (SDE) such that at every instant in pseudo-time each one is distributed according to the appropriate density in the sequence \eqref{eq:density_sequence},
\begin{IEEEeqnarray}{rCl}
 d\ls{\pt} & = & \flowdrift{\pt}(\ls{\pt}) d\pt + \flowdiffuse{\pt}(\ls{\pt}) d\flowbm{\pt} \label{eq:state_sde}     ,
\end{IEEEeqnarray}
in which $\flowdrift{\pt}$ and $\flowdiffuse{\pt}$ are drift and diffusion terms, and $\flowbm{\pt}$ is Brownian motion.

At the end, since $\seqden{1} = \postden$, the final particles are independent and identically distributed according to the posterior. Hence, from the basic Monte Carlo principle, they may be used to form a consistent estimator of posterior expectations akin to \eqref{eq:consistent_estimator} but with uniform weights (i.e. $\pw{}\pss{i}=1$).

The challenge in applying such a particle flow sampler comes in finding suitable dynamics with which to move the particles such that the correct density is maintained throughout. In general, this cannot be achieved analytically, and approximations are called for (see aforesaid references). While these may sometimes lead to effective estimators, they result in the loss of consistency, and the introduction of asymptotic bias which is not easily quantified.

\subsection{Exact Particle Flows}

It may be shown that exact particle flows obey the following governing equation.
\begin{theorem} \label{theo:flow_governing_equation}
For a particle moving according to \eqref{eq:state_sde}, if the drift and diffusion are differentiable and satisfy,
\begin{IEEEeqnarray}{l}
 \loglhood(\ls{\pt}) - \expect{\seqden{\pt}}\left[ \loglhood \right] + \trace\left[ \pdv{\flowdrift{\pt}}{\ls{\pt}} \right] + \pdv{\logseqden{\pt}}{\ls{\pt}}^T \flowdrift{\pt}(\ls{\pt}) - \trace\left[ \flowcov{\pt}(\ls{\pt}) \ppdv{\logseqden{\pt}}{\ls{\pt}} \right]  \nonumber \\
 \qquad\qquad\qquad -\: \pdv{\logseqden{\pt}}{\ls{\pt}}^T \flowcov{\pt}(\ls{\pt}) \pdv{\logseqden{\pt}}{\ls{\pt}} - 2 \sum_{ij} \pdv{\flowcov{\pt,ij}}{\ls{\pt,i}} \pdv{\logseqden{\pt}}{\ls{\pt,j}} - \sum_{ij} \mpdv{\flowcov{\pt,ij}}{\ls{\pt,i}}{\ls{\pt,j}} = 0 \label{eq:optimal_flow_pde}         ,
\end{IEEEeqnarray}
in which
\begin{IEEEeqnarray}{rClCrCl}
 \logseqden{\pt}(\ls{}) & = & \log(\seqden{\pt}(\ls{})) & \qquad\qquad & \loglhood(\ls{}) & = & \log(\lhood(\ls{}))  \nonumber \\
 \flowcov{\pt}(\ls{})   & = & \half \flowdiffuse{\pt}(\ls{}) \flowdiffuse{\pt}(\ls{})^T & \qquad\qquad & \expect{\seqden{\pt}}\left[ \loglhood \right] & = & \int \seqden{\pt}(\ls{}) \loglhood(\ls{}) d\ls{} \label{eq:optimal_flow_pde_terms}      ,
\end{IEEEeqnarray}
then the marginal density of $\ls{\pt}$ is $\seqden{\pt}(\ls{})$ as defined by \eqref{eq:density_sequence}. For proof see appendix~\ref{app:governing_equation} which is based on \cite{Daum2008}.
\end{theorem}

The governing equation relates the SDE drift and diffusion to three quantities: the gradient $\pdv{\logseqden{\pt}}{\ls{\pt}}$ and Hessian $\ppdv{\logseqden{\pt}}{\ls{\pt}}$ of the log-density at the current location, and the expected value of the log-likelihood $\expect{\seqden{\pt}}\left[ \loglhood \right]$. Intuitively, the derivative terms may be seen as controlling the particle motion due to changes in the local shape of the sequence density, while the expectation controls motion due to shifts in the bulk of the probability mass.

\subsection{Particle Flow Importance Sampling}

Since \eqref{eq:optimal_flow_pde} cannot in general be solved, the approach adopted in this paper is to combine particle flow with importance sampling. Suppose that each one of a collection of particles moves according to an SDE \eqref{eq:state_sde}, but where $\flowdrift{\pt}$ and $\flowdiffuse{\pt}$ do not satisfy \eqref{eq:optimal_flow_pde}, and the resulting density for $\ls{\pt}$ is $\partden{\pt}(\ls{}) \ne \seqden{\pt}(\ls{})$. The ideal importance weight is then simply,
\begin{IEEEeqnarray}{rCl}
 \pw{\pt} & = & \frac{ \seqden{\pt}(\ls{\pt}) }{ \partden{\pt}(\ls{\pt}) } \label{eq:ideal_weight}      .
\end{IEEEeqnarray}
We can establish a differential equation for this weight.
\begin{theorem} \label{theo:ideal_weight}
A collection of particles moving according to \eqref{eq:state_sde} with differentiable drift and diffusion, and with log-weights $\logpw{\pt} = \log(\pw{\pt})$ evolving according to,
\begin{IEEEeqnarray}{rCl}
 d\logpw{\pt} & = & \Bigg\{ \loglhood(\ls{\pt}) - \expect{\seqden{\pt}}\left[ \loglhood \right] + \trace\left[\pdv{\flowdrift{\pt}}{\ls{\pt}}\right] + \pdv{\logseqden{\pt}}{\ls{\pt}}^T\flowdrift{\pt}(\ls{\pt}) + \trace\left[ \flowcov{\pt}(\ls{\pt}) \ppdv{\logseqden{\pt}}{\ls{\pt}} \right] \nonumber \\
 & & \qquad -\: 2 \trace\left[ \flowcov{\pt}(\ls{\pt}) \ppdv{\logpartden{\pt}}{\ls{\pt}} \right] - \pdv{\logpartden{\pt}}{\ls{\pt}}^T \flowcov{\pt}(\ls{\pt}) \pdv{\logpartden{\pt}}{\ls{\pt}} - 2 \sum_{ij} \pdv{\flowcov{\pt,ij}}{\ls{\pt,i}} \pdv{\logpartden{\pt}}{\ls{\pt,j}} \nonumber \\
 & & \qquad -\: \sum_{ij} \mpdv{\flowcov{\pt,ij}}{\ls{\pt,i}}{\ls{\pt,j}} \Bigg\} d\pt + \left[ \pdv{\logseqden{\pt}}{\ls{\pt}} - \pdv{\logpartden{\pt}}{\ls{\pt}} \right]^T \flowdiffuse{\pt}(\ls{\pt}) d\flowbm{\pt} \nonumber \\ \label{eq:ideal_weight_differential_equation} 
\end{IEEEeqnarray}
in which
\begin{IEEEeqnarray}{rClCrCl}
 \logpartden{\pt}(\ls{}) & = & \log(\partden{\pt}(\ls{})) & \qquad & \logseqden{\pt}(\ls{}) & = & \log(\seqden{\pt}(\ls{})) \nonumber     ,
\end{IEEEeqnarray}
is properly weighted with respect to $\seqden{\pt}$, and the resulting weights correspond to \eqref{eq:ideal_weight}.
For proof see appendix~\ref{app:ideal_weight}
\end{theorem}
When $\logpartden{\pt} = \logseqden{\pt}$, and $\flowdrift{\pt}$ and $\flowdiffuse{\pt}$ satisfy the conditions of theorem~\ref{theo:flow_governing_equation}, then it is clear from \eqref{eq:optimal_flow_pde} that $d\logpw{\pt}=0$, as we would expect when simulating perfectly from the target distribution. Furthermore, we can omit the term $\expect{\seqden{\pt}}\left[ \loglhood \right]$ in calculation since this does not depend on $\ls{\pt}$ and thus will cancel out when the final weights are normalised. This is equivalent to using the unnormalised target density in \eqref{eq:ideal_weight}.

If it were possible to simulate particle motion according to a chosen SDE, and at the same time evaluate the corresponding weights using \eqref{eq:ideal_weight_differential_equation}, then the particles would be properly weighted importance samples for all $\pt\in[0,1]$, and standard convergence results would apply. In practice, it will be necessary to use approximate numerical integration schemes. Provided that these recover the ideal continuous-time evolution of both the particle state and weight as the step size tends to zero, then the resulting particle flow importance sampling will retain these asymptotic properties, but only in the limit as both the step sizes to go zero and the number of particles to infinity.

In practice, designing a numerical integration scheme which approximates \eqref{eq:ideal_weight_differential_equation} is not possible, because of the dependence on the unknown $\logpartden{\pt}$, apart from in the special case when $\flowdiffuse{\pt}(\ls{\pt})=0$. Instead we show that there are other valid ways in which the weight may evolve which still result in a properly weighted collection of particles. These use the concept of targeting an extended distribution over a larger set of variables for the importance sampling.

\begin{theorem}\label{theo:practical_weight}
A collection of particles moving according to \eqref{eq:state_sde} with differentiable drift and diffusion, and with log-weights $\logpw{\pt} = \log(\pw{\pt})$ evolving according to,
%
\begin{IEEEeqnarray}{rCl}
 d\logpw{\pt} & = & \Bigg\{ \loglhood(\ls{\pt}) - \expect{\seqden{\pt}}\left[ \loglhood \right] + \trace\left[\pdv{\flowdrift{\pt}}{\ls{\pt}}\right] + \pdv{\logseqden{\pt}}{\ls{\pt}}^T\flowdrift{\pt}(\ls{\pt}) + \trace\left[ \flowcov{\pt}(\ls{\pt}) \npdv{2}{\logseqden{\pt}}{\ls{\pt}} \right] \nonumber \\
 & & \qquad\qquad -\: \half \sum_{ijk} \left[ \pdv{\flowdiffuse{\pt,ik}}{\ls{\pt,j}} \pdv{\flowdiffuse{\pt,jk}}{\ls{\pt,i}} \right] \Bigg\} d\pt + \sum_{ij} \pdv{\flowdiffuse{\pt,ij}}{\ls{\pt,i}} d\flowbm{\pt,j} + \pdv{\logseqden{\pt}}{\ls{\pt}}^T \flowdiffuse{\pt} d\flowbm{\pt} \label{eq:practical_weight_differential_equation} 
\end{IEEEeqnarray}
is properly weighted with respect to $\seqden{\pt}$. The proof uses the construction of an extended target distribution which encompasses the path of the Brownian motion. See appendix~\ref{app:practical_weight}.
\end{theorem}

Notice that \eqref{eq:practical_weight_differential_equation} coincides with the ideal form \eqref{eq:ideal_weight_differential_equation} when $\flowdiffuse{\pt}(\ls{\pt})=0$. It is possible to construct a suitable numerical scheme which results in weight evolution according to \eqref{eq:practical_weight_differential_equation} as the step size tends to $0$.

The idea of combining particle and importance sampling is somewhat in the spirit of \citep{Reich2013}, but the construction used here is substantially different, both in the type of particle flow employed and in the calculation of the weights.

\section{Sampling with Gaussian Flows} \label{sec:gaussian_flows}

\subsection{Exact Gaussian Flows for Linear Gaussian Models}


When the model used is linear and Gaussian, the exact flow for particle motion may be derived analytically. Suppose the likelihood takes the form of an observation $\ob{}\in\obspace = \real^{\obdim}$ which is linearly dependent on the state with Gaussian noise, and that the prior is also Gaussian, as follows.
\begin{model} \label{mod:linear_gaussian}
\begin{IEEEeqnarray}{rClCrCl}
 \priorden(\ls{}) & = & \normalden{\ls{}}{\lsmn{0}}{\lsvr{0}} & \qquad\qquad & \lhood(\ls{}) & = & \normalden{\ob{}}{\lgmom\ls{}}{\lgmov}
\end{IEEEeqnarray}
$\lsvr{0}$ and $\lgmov$ are positive definite covariance matrices.
\end{model}
The following properties may be established.
\begin{proposition} \label{prop:linear_gaussian_density_sequence}
For model~\ref{mod:linear_gaussian}, the geometric density sequence \eqref{eq:density_sequence} is,
\begin{IEEEeqnarray}{rCl}
 \seqden{\pt}(\ls{}) & = & \normalden{\ls{}}{\lsmn{\pt}}{\lsvr{\pt}} \label{eq:linear_gaussian_density_sequence}
\end{IEEEeqnarray}
\begin{IEEEeqnarray}{rClCrCl}
 \lsvr{\pt} & = & \left(\lsvr{0}^{-1} + \pt \lgmom^T \lgmov^{-1} \lgmom\right)^{-1} & \qquad \lsmn{\pt} & = & \lsvr{\pt} \left[ \lsvr{0}^{-1} \lsmn{0} + \pt \lgmom^T \lgmov^{-1} \ob{} \right] \label{eq:gaussian_mean_variance}      .
\end{IEEEeqnarray}
\end{proposition}

\begin{proof}
The proof is straightforward using standard identities for Gaussian densities.\qed
\end{proof}

\begin{proposition} \label{theo:optimal_gaussian_flow}
A particle sampled from the prior of model~\ref{mod:linear_gaussian} and moved according to SDE \eqref{eq:state_sde} over the interval $[0,1]$ follows the density $\seqden{\pt}$ defined in \eqref{eq:linear_gaussian_density_sequence} when the drift and diffusion terms are set to,
\begin{IEEEeqnarray}{rCl}
 \flowdrift{\pt}(\ls{\pt}) & = & \lsvr{\pt} \lgmom^T \lgmov^{-1} \left( \left(\ob{} - \lgmom \lsmn{\pt} \right) - \half \lgmom (\ls{\pt}-\lsmn{\pt}) \right) - \half \dsf (\ls{\pt}-\lsmn{\pt}) \nonumber \\
 \flowdiffuse{\pt}         & = & \dsf^{\half} \lsvr{\pt}^{\half} \label{eq:gaussian_flow_drift_diffusion}      ,
\end{IEEEeqnarray}
where $\dsf \ge 0$ is a design parameter of the flow.
\end{proposition}

\begin{proof}
Substituting in $\flowdrift{\pt}$ and $\flowdiffuse{\pt}$ from \eqref{eq:gaussian_flow_drift_diffusion}, it is immediately clear that the governing equation \eqref{eq:optimal_flow_pde} is satisfied, and hence that the flow is exact.\qed
\end{proof}

The behaviour of the state dynamics is controlled through the choice of $\dsf$. When $\dsf=0$, the particle motion is deterministic; when $\dsf>0$, stochastic.

\begin{theorem} \label{theo:integrated_gaussian_flow}
For model~\ref{mod:linear_gaussian}, a particle $\ls{\pt_0}\sim\seqden{\pt_0}$ with this density defined in \eqref{eq:linear_gaussian_density_sequence} and moved according to SDE \eqref{eq:state_sde} over the interval $[\pt_0,\pt_1]$  with drift and diffusion terms as in \eqref{eq:gaussian_flow_drift_diffusion} reaches the state,
\begin{IEEEeqnarray}{rCl}
 \ls{\pt_1} & = & \lsmn{\pt_1} + \exp\left\{-\half\dsf(\pt_1-\pt_0)\right\} \left(\lsvr{\pt_1}\lsvr{\pt_0}^{-1}\right)^{\half}(\ls{\pt_0}-\lsmn{\pt_0}) \nonumber \\
 & & \qquad \qquad \qquad \qquad \qquad + \: \left[ \frac{1-\exp\left\{-\dsf(\pt_1-\pt_0)\right\}}{\pt_1-\pt_0} \right]^{\half} \lsvr{\pt_1}^{\half} \left(\flowbm{\pt_1}-\flowbm{\pt_0}\right) \label{eq:state_update}     ,
\end{IEEEeqnarray}
where $A^{\half}$ is the principal matrix square root of $A$.
\end{theorem}

\begin{proof}
A constructive proof is possible by solving the SDE. This may be accomplished using a matrix integrating factor approach, and is rather lengthy. Having obtained the solution, it is straightforward to verify that it satisfies the SDE. See appendix \ref{app:integrated_gaussian_flow}.\qed
\end{proof}

Using equation~\eqref{eq:state_update}, it is possible to calculate or sample the state at any point in pseudo-time given the state at some earlier point in pseudo-time. An example is shown in figure~\ref{fig:gaussian_flow_example}.

\begin{figure}[bt]
\centering
\subfloat[]{ \includegraphics{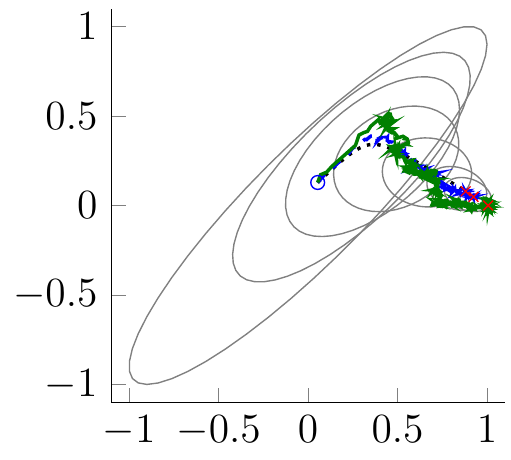} }
\subfloat[]{ \includegraphics{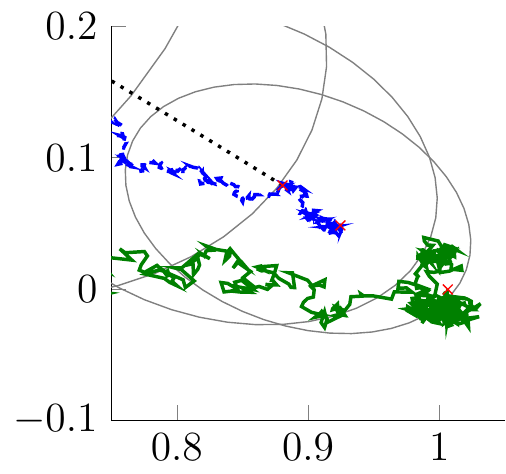} }
\captionsetup{singlelinecheck=off}
\caption[.]{An illustration of a Gaussian flow for a linear Gaussian model. The ellipses are one standard deviation contours of a selection of the sequence densities. The paths show the evolution of three particles from the same starting state using $\dsf=0$ (dotted), $\dsf=0.03$ (dashed) and $\dsf=0.3$ (solid). The initial, prior-sampled state is shown with a circle. The second panel shows a detailed view of the final stages of the trajectories. Model parameters are $\lsmn{0}=\begin{bmatrix}0 & 0\end{bmatrix}^T$, $\lsvr{0}=\begin{bmatrix}1 & 0.9 \\ 0.9 & 1\end{bmatrix}$, $\ob{}=\begin{bmatrix}1 & 0\end{bmatrix}^T$, $\lgmom=I$, $\lgmov=\begin{bmatrix}0.02 & 0.005 \\ 0.005 & 0.01\end{bmatrix}$.}
\label{fig:gaussian_flow_example}
\end{figure}

\subsection{Approximate Gaussian Flows for Nonlinear Gaussian Models} \label{sec:nonlinear_gaussian_models}

For the linear Gaussian models of the previous section, sampling using a particle flow is clearly of no practical use, since the posterior distribution may be computed and sampled directly. However, it may be used as the basis of an approximately optimal flow for less tractable models. Consider the class of models with Gaussian densities but with a nonlinear dependence of the observation on the state. (N.B. In a filtering setting this encompasses the common case where both the transition and observation functions are nonlinear with additive Gaussian noise.)
\begin{model} \label{mod:nonlinear_gaussian}
\begin{IEEEeqnarray}{rClCrCl}
 \priorden(\ls{}) & = & \normalden{\ls{}}{\lsmn{0}}{\lsvr{0}} & \qquad\qquad & \lhood(\ls{}) & = & \normalden{\ob{}}{\obsfun(\ls{})}{\lgmov}
\end{IEEEeqnarray}
The observation function $\obsfun$ is twice differentiable with respect to $\ls{}$.
\end{model}

For nonlinear Gaussian models, the density sequence is not available analytically, nor is there a closed form expression for the particle flow. However, we can initialise the flow exactly with a sample from the Gaussian prior, and then approximate the optimal dynamics using the Gaussian flow defined in proposition~\ref{theo:optimal_gaussian_flow}. The key to this approximation is to linearise the likelihood using a truncated Taylor expansion,
\begin{IEEEeqnarray}{rClCl}
 \lhood(\ls{}) & \approx & \lhoodapprox{}(\ls{} \cent \lslin) & = & \normalden{\obapprox{}(\lslin)}{\lgmomapprox{}(\lslin)\ls{}}{\lgmov} \nonumber 
\end{IEEEeqnarray}
\begin{IEEEeqnarray}{rClCrCl}
 \lgmomapprox{}(\lslin) & = & \pd{\obsfun}{\ls{}}{\lslin} & \qquad & \obapprox{}(\lslin) & = & \ob{} - \obsfun(\lslin) + \lgmomapprox{}(\lslin) \lslin \label{eq:linearisation}      .
\end{IEEEeqnarray}
Using the linearised model, we can write the approximate Gaussian moments,
\begin{IEEEeqnarray}{rCl}
 \lsvrapprox{\pt}(\lslin) & = & \left(\lsvr{0}^{-1} + \pt \lgmomapprox{}(\lslin)^T \lgmov^{-1} \lgmomapprox{}(\lslin)\right)^{-1} \nonumber \\
 \lsmnapprox{\pt}(\lslin) & = & \lsvrapprox{\pt}(\lslin) \left[ \lsvr{0}^{-1} \lsmn{0} + \pt \lgmomapprox{}(\lslin)^T \lgmov^{-1} \obapprox{}(\lslin) \right] \label{eq:approx_gaussian_mean_variance}      ,
\end{IEEEeqnarray}
which parameterise the following density sequence,
\begin{IEEEeqnarray}{rClCrCl}
 \seqdenapprox{\pt|\lslin}(\ls{}) & = & \normalden{\ls{}}{\lsmnapprox{\pt}(\lslin)}{\lsvrapprox{\pt}(\lslin)} & \qquad & \logseqdenapprox{\pt|\lslin}(\ls{}) = \log\left(\seqdenapprox{\pt|\lslin}(\ls{})\right) \nonumber      .
\end{IEEEeqnarray}
For this linear Gaussian approximation, the exact drift and diffusion are,
\begin{IEEEeqnarray}{rCl}
 \flowdriftapprox{\pt}(\ls{\pt} \cent \lslin) & = & \lsvrapprox{\pt}(\lslin) \lgmomapprox{}(\lslin)^T \lgmov^{-1} \left( \left(\obapprox{}(\lslin) - \lgmomapprox{}(\lslin) \lsmnapprox{\pt}(\lslin) \right) - \half \lgmomapprox{}(\lslin) (\ls{\pt}-\lsmnapprox{\pt}(\lslin)) \right) \nonumber \\
 & & \qquad\qquad\qquad\qquad\qquad\qquad\qquad\qquad\qquad\qquad -\: \half \dsf (\ls{\pt}-\lsmnapprox{\pt}(\lslin)) \nonumber \\
 \flowdiffuseapprox{\pt}(\lslin) & = & \dsf^{\half} \lsvrapprox{\pt}(\lslin)^{\half} \label{eq:approx_gaussian_flow_drift_diffusion}      .
\end{IEEEeqnarray}
We define the \emph{approximate Gaussian flow} using these expressions, with the linearisation conducted around the current state,
\begin{IEEEeqnarray}{rCl}
 d\ls{\pt} & = & \flowdriftapprox{\pt}(\ls{\pt} \cent \ls{\pt}) d\pt + \flowdiffuseapprox{\pt}(\ls{\pt}) d\flowbm{\pt} 
\end{IEEEeqnarray}
The choice of $\lslin=\ls{\pt}$ ensures that $\lhoodapprox{}(\ls{\pt} \: ; \: \ls{\pt}) = \lhood(\ls{\pt})$, and similar equivalence for the derivatives. Using the governing equation for exact particle flows and considering both $\seqden{\pt}(\ls{})$ and $\seqdenapprox{\pt|\ls{\pt}}(\ls{})$, it is then straightforward to show from \eqref{eq:optimal_flow_pde} that the approximate Gaussian flow will be optimal if,
\begin{IEEEeqnarray}{rCl}
 \expect{\seqden{\pt}}\left[ \loglhood \right] - \expect{\seqdenapprox{\pt|\ls{\pt}}}\left[ \loglhood \right] + \text{terms involving $\frac{d^2\obsfun}{d\ls{\pt}^2}$} & = & 0     .
\end{IEEEeqnarray}
where $\frac{d^2\obsfun}{d\ls{\pt}^2}$ is the tensor of second derivatives of the observation function.

If $\frac{d^2\obsfun}{d\ls{}^2}=0$ for all $\ls{}$, then the model is linear and Gaussian and we recover the exact Gaussian flow. However, the flow is still optimal in the more general case where $\frac{d^2\obsfun}{d\ls{\pt}^2}=0$ only at the current state $\ls{\pt}$, and also $\expect{\seqden{\pt}}\left[ \loglhood \right] - \expect{\seqdenapprox{\pt|\ls{\pt}}}\left[ \loglhood \right]=0$. Hence, the use of an approximate Gaussian flow implies two assumptions, that the second derivatives are small along the particle trajectory, and that the expected log-likelihood can be well-approximated using a Gaussian density.

An illustration of approximate Gaussian flow is shown in figure~\ref{approx_gaussian_flow_example}.

\begin{figure}[bt]
\centering
\subfloat[]{ \includegraphics{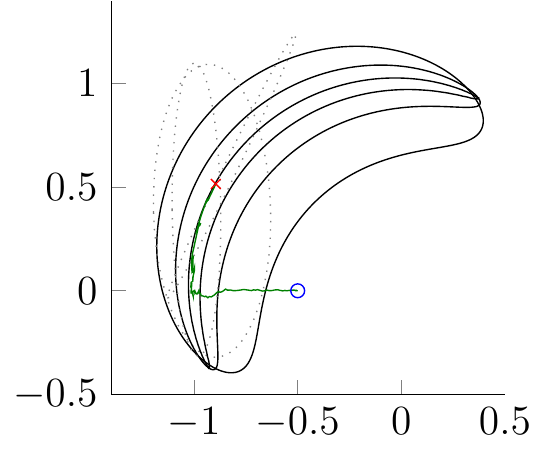} }
\subfloat[]{ \includegraphics{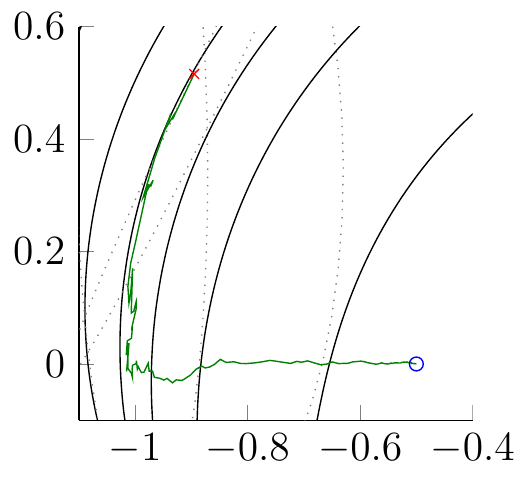} }
\captionsetup{singlelinecheck=off}
\caption[.]{An illustration of an approximate Gaussian flow for a nonlinear Gaussian model. Solid contours represent the evolution of the true density sequence, and dotted contours those of the Gaussian approximations at the same times. The resulting path of a particle using $\dsf=0.1$ is also shown. The second panel shows a close-up view of the trajectory. Model parameters are $\lsmn{0}=\begin{bmatrix}-0.4 & -0.4\end{bmatrix}^T$, $\lsvr{0}=0.5 I$, $\ob{}=1$, $\lgmov=0.001$, $\obsfun(\ls{})=\sqrt{\ls{1}^2+\ls{2}^2}$.}
\label{approx_gaussian_flow_example}
\end{figure}

\subsection{Numerical Integration of the Approximate Gaussian Flow}

To implement the particle flow importance sampling algorithm, we need a numerical integration scheme which will allow us to approximately sample a joint trajectory for each particle state and its associated importance weight. This could be achieved using the Euler method. However, a more accurate option is available to us which exploits the analytical solution to the flow for the linear Gaussian model.

\subsubsection{State Integration}
The SDE for the approximate Gaussian flow can be written in the following form,
\begin{IEEEeqnarray}{rCl}
 d\ls{\pt} & = & \left[\flowdriftscale{\pt}(\ls{\pt}) \ls{\pt} + \flowdriftconst{\pt}(\ls{\pt}) \right] d\pt + \flowdiffuseapprox{\pt}(\ls{\pt}) d\flowbm{\pt} \label{eq:state_sde_agf}     ,
\end{IEEEeqnarray}
where
\begin{IEEEeqnarray}{rCl}
 \flowdriftscale{\pt}(\ls{\pt}) & = & - \half \lsvrapprox{\pt}(\ls{\pt}) \lgmomapprox{}(\ls{\pt})^T \lgmov^{-1} \lgmomapprox{}(\ls{\pt}) - \half \dsf I \nonumber \\
 \flowdriftconst{\pt}(\ls{\pt}) & = & \lsvrapprox{\pt}(\ls{\pt}) \lgmomapprox{}(\ls{\pt})^T \lgmov^{-1} \left(\obapprox{}(\ls{\pt}) - \half \lgmomapprox{}(\ls{\pt}) \lsmnapprox{\pt}(\ls{\pt}) \right) + \half \dsf \lsmnapprox{\pt}(\ls{\pt}) \nonumber      .
\end{IEEEeqnarray}
For an integration step from $\pt_0$ to $\pt_1$, the Euler method provides an approximate value of $\ls{\pt_1}$ which is an exact solution to the SDE with terms fixed at $\pt_0$ and $\ls{\pt_0}$,
\begin{IEEEeqnarray}{rCl}
 d\ls{\pt} & = & \left[\flowdriftscale{\pt_0}(\ls{\pt_0}) \ls{\pt_0} + \flowdriftconst{\pt_0}(\ls{\pt_0})\right] d\pt + \flowdiffuseapprox{\pt_0}(\ls{\pt_0}) d\flowbm{\pt} \label{eq:state_sde_agf_euler}     .
\end{IEEEeqnarray}
As the integration step size goes to $0$, trajectories generated using the Euler method become exact samples according to the true SDE. For our flow, we can instead fix only the state used to form the linear approximation,
\begin{IEEEeqnarray}{rCl}
 d\ls{\pt} & = & \left[\flowdriftscale{\pt}(\ls{\pt_0}) \ls{\pt} + \flowdriftconst{\pt}(\ls{\pt_0}) \right] d\pt + \flowdiffuseapprox{\pt}(\ls{\pt_0}) d\flowbm{\pt} \label{eq:state_sde_agf_gaussian}     .
\end{IEEEeqnarray}
This implies matching some additional $\bigo{\dpt}$ and $\bigo{\dls}$ terms in the Taylor expansions of $\flowdriftapprox{\pt}(\ls{\pt})$ and $\flowdiffuseapprox{\pt}(\ls{\pt})$. Since the resulting SDE describes a Gaussian flow, it may be solved exactly using \eqref{eq:state_update}, leading in this case to,
\begin{IEEEeqnarray}{rCl}
 \ls{\pt_1} & = & \lsmnapprox{\pt_1}(\ls{\pt_0}) + \exp\left\{-\half\dsf(\pt_1-\pt_0)\right\} \left(\lsvrapprox{\pt_1}(\ls{\pt_0})\lsvrapprox{\pt_0}(\ls{\pt_0})^{-1}\right)^{\half}(\ls{\pt_0}-\lsmnapprox{\pt_0}(\ls{\pt_0})) \nonumber \\
 & & \qquad \qquad \qquad \qquad + \: \left[ \frac{1-\exp\left\{-\dsf(\pt_1-\pt_0)\right\}}{\pt_1-\pt_0} \right]^{\half} \lsvrapprox{\pt_1}(\ls{\pt_0})^{\half} \left(\flowbm{\pt_1}-\flowbm{\pt_0}\right) \label{eq:approx_state_update}     .
\end{IEEEeqnarray}
From theorem~\ref{theo:integrated_gaussian_flow}, as the step size goes to $0$, this recovers the ideal continuous time behaviour for the approximate Gaussian flow.

\subsubsection{Weight Integration}
A corresponding approximate weight update may be conducted by conditioning on the sampled value of $\left(\flowbm{\pt_1}-\flowbm{\pt_0}\right)$ and using the Jacobian of \eqref{eq:approx_state_update},
\begin{IEEEeqnarray}{rCl}
 \pw{\pt_1} & = & \pw{\pt_0} \times \frac{ \seqden{\pt_1}(\ls{\pt_1}) }{ \seqden{\pt_0}(\ls{\pt_0}) } \times \determ{ \pdv{\ls{\pt_1}}{\ls{\pt_0}} } \nonumber \\
 & \propto & \pw{\pt_0} \times \frac{ \priorden(\ls{\pt_1}) \lhood(\ls{\pt_1})^{\pt_1} }{ \priorden(\ls{\pt_0}) \lhood(\ls{\pt_0})^{\pt_0} } \times \determ{ \pdv{\ls{\pt_1}}{\ls{\pt_0}} } \label{eq:approx_weight_update}      .
\end{IEEEeqnarray}
This is approximate, in the sense that it does not result in properly weighted samples, because the state update is not in general an invertible function due to the nonlinear dependence on $\ls{\pt_0}$. However, we can establish the following result.

\begin{theorem} \label{theo:weight_numerical_integration}
 A particle with state $\ls{\pt}$ moved according to \eqref{eq:approx_state_update} and weight $\pw{\pt}$ according to \eqref{eq:approx_weight_update}, is properly weighted according to $\seqden{\pt}$ as the integration step size goes to $0$. For proof see appendix~\ref{app:weight_numerical_integration}.
\end{theorem}

Using the chain rule, the $(i,j)$th element of the Jacobian matrix is,
\begin{IEEEeqnarray}{rCl}
 \left[\pdv{\ls{\pt_1}}{\ls{\pt_0}}\right]_{i,j} & = & \left[ \pdv{\lsmnapprox{\pt_1}}{\ls{\pt_0}} + \exp\left\{-\half\dsf(\pt_1-\pt_0)\right\} \left(\lsvrapprox{\pt_1}(\ls{\pt_0})\lsvrapprox{\pt_0}(\ls{\pt_0})^{-1}\right)^{\half}\left(I-\pdv{\lsmnapprox{\pt_0}}{\ls{\pt_0}}\right) \right]_{i,j} \nonumber \\
 & & + \: \left( \frac{1-\exp\left\{-\dsf(\pt_1-\pt_0)\right\}}{\pt_1-\pt_0} \right)^{\half} \sum_k \left[\pdv{\lsvrapprox{\pt_1}^{\half}}{\ls{\pt_0,j}}\right]_{i,k} \left[\flowbm{\pt_1}-\flowbm{\pt_0}\right]_k \nonumber \\
 & & + \: \exp\left\{-\half\dsf(\pt_1-\pt_0)\right\} \sum_k \left[\pdv{\left(\lsvrapprox{\pt_1}\lsvrapprox{\pt_0}^{-1}\right)^{\half}}{\ls{\pt_0,j}}\right]_{i,k} (\ls{\pt_0}-\lsmnapprox{\pt_0}) 
\end{IEEEeqnarray}
The $(j)$th column of the derivative of $\lsmnapprox{\pt}(\ls{\pt_0})$ is given by,
\begin{IEEEeqnarray}{rCl}
 \pdv{\lsmnapprox{\pt}}{\ls{\pt_0,j}} & = & \pt \lsvrapprox{\pt}(\ls{\pt_0}) \Bigg( \pdv{\lgmomapprox{}^T}{\ls{\pt_0,j}} \lgmov^{-1}\left(\obapprox{}(\ls{\pt_0})-\lgmomapprox{}(\ls{\pt_0})\lsmnapprox{\pt}(\ls{\pt_0})\right) \nonumber \\
 & & \qquad\qquad\qquad\qquad\qquad +\: \lgmomapprox{}(\ls{\pt_0})^T\lgmov^{-1}\pdv{\lgmomapprox{}}{\ls{\pt_0,j}}\left(\ls{\pt_0}-\lsmnapprox{\pt}(\ls{\pt_0})\right) \Bigg)     ,
\end{IEEEeqnarray}
where $\pdv{\lgmomapprox{}^T}{\ls{\pt_0,j}}$ is a matrix whose $(i,k)$th term is $\mpdv{\obsfun_i}{\ls{\pt_0,j}}{\ls{\pt_0,k}}$. The two matrix square root derivatives may be evaluated by observing that since $A^{\half}A^{\half}=A$, then,
\begin{IEEEeqnarray}{rCl}
 A^{\half} \pdv{A^{\half}}{\ls{\pt_0,j}} + \pdv{A^{\half}}{\ls{\pt_0,j}} A^{\half} = \pdv{A}{\ls{\pt_0,j}} \nonumber     ,
\end{IEEEeqnarray}
and hence the derivative may be found by solving a Sylvester equation by standard methods \citep{Bartels1972}. On the right hand side of these equations we need,
\begin{IEEEeqnarray}{rCl}
 \pdv{\lsvrapprox{\pt_1}}{\ls{\pt_0,j}} & = & -\pt_1 \lsvrapprox{\pt_1}(\ls{\pt_0}) \left( \pdv{\lgmomapprox{}^T}{\ls{\pt_0,j}}\lgmov^{-1}\lgmomapprox{}(\ls{\pt_0}) + \lgmomapprox{}(\ls{\pt_0})\lgmov^{-1}\pdv{\lgmomapprox{}^T}{\ls{\pt_0,j}} \right) \lsvrapprox{\pt_1}(\ls{\pt_0}) \nonumber \\
 \pdv{\left(\lsvrapprox{\pt_1}\lsvrapprox{\pt_0}^{-1}\right)}{\ls{\pt_0,j}} & = & \lsvrapprox{\pt_1}(\ls{\pt_0}) \left( \pdv{\lgmomapprox{}^T}{\ls{\pt_0,j}}\lgmov^{-1}\lgmomapprox{}(\ls{\pt_0}) + \lgmomapprox{}(\ls{\pt_0})\lgmov^{-1}\pdv{\lgmomapprox{}^T}{\ls{\pt_0,j}} \right) \nonumber \\
 & & \qquad\qquad\qquad\qquad \times \left(\pt_0 I - \pt_1 \lsvrapprox{\pt_1}(\ls{\pt_0})\lsvrapprox{\pt_0}(\ls{\pt_0})^{-1}\right) \nonumber     .
\end{IEEEeqnarray}

\subsubsection{Step Size Control} \label{sec:step_size}
Effective numerical integration requires careful consideration of the integration step sizes. Using smaller step sizes reduces the associated errors, resulting in a path which more accurately represents a sample from the approximate Gaussian flow. However, in practice the number of steps needs to be kept fairly low, to minimise the computational burden. In some instances, it may be sufficient to use a fixed step size, or a predetermined time grid chosen with a tuning run. However, an adaptive scheme is preferable for greatest efficiency.

Adaptation may be conducted by forming an estimate of the error introduced by the numerical integration scheme, and adjusting the step size to keep this below a predetermined threshold. When $\dsf>0$ and the flow is stochastic, the error is calculated conditional on the sampled path of the Brownian motion. Step size adaptation is performed independently for each particle.

To form a local error estimate for each integration step, we can compare the SDE for the approximate Gaussian flow \eqref{eq:state_sde_agf} with that implied by the numerical integration \eqref{eq:state_sde_agf_gaussian}. The difference between them describes the evolution of the integration error,
\begin{IEEEeqnarray}{rCl}
 d\errstat{\pt} & = & \left[ \left( \flowdriftscale{\pt}(\ls{\pt})-\flowdriftscale{\pt}(\ls{\pt_0}) \right) \ls{\pt} + \left(\flowdriftconst{\pt}(\ls{\pt})-\flowdriftconst{\pt}(\ls{\pt_0})\right) \right] + \left[ \flowdiffuseapprox{\pt}(\ls{\pt})-\flowdiffuseapprox{\pt}(\ls{\pt_0}) \right] d\flowbm{\pt} \nonumber      .
\end{IEEEeqnarray}
Integrating from $\pt_0$ to $\pt_1$ and approximating each integrand with the average of its initial and final value,
\begin{IEEEeqnarray}{rCl}
 \errstat{\pt_1|\pt_0} & = & \int_{\pt_0}^{\pt_1} \left[ \left( \flowdriftscale{\pt}(\ls{\pt})-\flowdriftscale{\pt}(\ls{\pt_0}) \right) \ls{\pt} + \left(\flowdriftconst{\pt}(\ls{\pt})-\flowdriftconst{\pt}(\ls{\pt_0})\right) \right] d\pt \nonumber \\
 & & \qquad\qquad\qquad\qquad\qquad +\: \int_{\pt_0}^{\pt_1} \left[ \flowdiffuseapprox{\pt}(\ls{\pt})-\flowdiffuseapprox{\pt}(\ls{\pt_0}) \right] d\flowbm{\pt} \nonumber \\
 & \approx & \half (\pt_1-\pt_0) \left[ \left( \flowdriftscale{\pt_1}(\ls{\pt_1})-\flowdriftscale{\pt_1}(\ls{\pt_0}) \right) \ls{\pt_1} + \left(\flowdriftconst{\pt_1}(\ls{\pt_1})-\flowdriftconst{\pt_0}(\ls{\pt_0})\right) \right] \nonumber \\
 & & \qquad\qquad\qquad\qquad\qquad +\: \half \left[ \flowdiffuseapprox{\pt_1}(\ls{\pt_1})-\flowdiffuseapprox{\pt_1}(\ls{\pt_0}) \right] (\flowbm{\pt_1}-\flowbm{\pt_0}) \label{eq:local_error_estimate}      .
\end{IEEEeqnarray}

Pseudo-time step sizes may now be adjusted so that the magnitude of this error statistic is kept below a threshold. If the error is too large then the new state is rejected and the step size reduced. Note than when $\dsf>0$ we must be particularly careful with this procedure. Since the weight and step size error calculations are conditioned on the sampled path of the Brownian motion $\flowbm{\pt}$, we cannot simply discard this and sample afresh. Intermediate values must be drawn from a Brownian bridge conditional on the existing skeleton of sampled points.

Mechanisms for adjusting the step sizes may be borrowed directly from well-established numerical integration algorithms for solving differential equations (see for example \citep{Shampine1997}).

Since $(\flowbm{\pt_1}-\flowbm{\pt_0})$ is of the order $\bigo{(\pt_1-\pt_0)^{\half}}$, it is expected that the stochastic term will dominate the integration error unless $\dsf$ is small. This suggests that we should ordinarily use $\dsf=0$. The advantage of using other values is in their use for implementing efficient MCMC kernels, as we discuss in section~\ref{sec:resample_move}.

\subsection{Summary}
Particle flow importance sampling may be conducted by first drawing a set of particles from the prior, and then allowing their states and weights to evolve according to an approximate Gaussian flow over an interval of pseudo time $\pt \in [0,1]$ using numerical integration. Pseudo-code for the procedure is provided in algorithm~\ref{alg:approx_gaussian_flow}.

\begin{algorithm}
\begin{spacing}{1.1}\footnotesize
\begin{algorithmic}[1]
 \REQUIRE Parameters $\lsmn{0}$, $\lsvr{0}$, $\lgmov$, $\ob{}$, $\obsfun$, according to model~\ref{mod:nonlinear_gaussian}.
 \FOR{$i=1$ \TO $\numpart$}
  \STATE Sample $\ls{0}\pss{i} \sim \normalden{\ls{}}{\lsmn{0}}{\lsvr{0}}$.
  \STATE Initialise unnormalised weight $\pw{0}\pss{i}=1$.
  \STATE Initialise Brownian motion $\flowbm{0}\pss{i}=0$.
  \STATE Set $\pt=0$
  \WHILE{$\pt\leq1$}
   \STATE Increment $\pt \leftarrow \min\left\{1, \: \pt+\Delta\pt\right\}$ with $\Delta\pt$ specified by a fixed grid or adaptive method. (See section~\ref{sec:step_size}).
   \STATE Linearise observation function using \eqref{eq:linearisation}.
   \STATE Sample a value for the Brownian motion $\flowbm{\pt}$.
   \STATE Calculate approximate Gaussian moments using \eqref{eq:approx_gaussian_mean_variance}.
   \STATE Advance state using \eqref{eq:approx_state_update}, yielding $\ls{\pt}\pss{i}$.
   \STATE Advance weight using \eqref{eq:approx_weight_update}, yielding $\pw{\pt}\pss{i}$.
  \ENDWHILE
  \STATE Set $\ls{}\pss{i} \leftarrow \ls{1}\pss{i}$ and $\pw{}\pss{i} \leftarrow \pw{1}\pss{i}$.
 \ENDFOR
 \STATE Normalise weights $\npw{}\pss{i} = \pw{}\pss{i} / \sum_j \pw{}\pss{j}$.
 \RETURN Importance weighted posterior samples $\left\{\ls{}\pss{i}, \: \npw{}\pss{i}\right\}$.
\end{algorithmic}
\caption{Gaussian flow importance sampling for the nonlinear Gaussian model.}
\label{alg:approx_gaussian_flow}
\end{spacing}
\end{algorithm}

\subsection{Performance Characterisation}

We now illustrate the the operation of particle flow importance sampling on a simple example model, and use this to explore its dependence on various model and algorithm parameters. The model is defined by the following parameters,
\begin{IEEEeqnarray}{c}
 \lsmn{0} = \mathbf{1} \qquad \lsvr{0}=\sigma_{x}^2 I \qquad \obsfun(\ls{})=\left({\sum_i \ls{i}^2}\right)^{\half} \qquad \lgmov=\sigma_{y}^2 \label{eq:test_model}     .
\end{IEEEeqnarray}
Figure~\ref{fig:gaussian_flow_importance_sampling} shows typical evolution of particle states and weights using a deterministic flow.

\begin{figure}[bt]
\centering
\subfloat[]{ \includegraphics{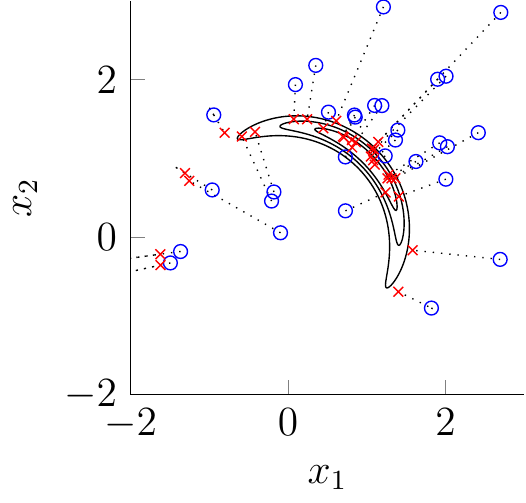} \label{subfig:gaussian_flow_importance_sampling_state}}
\subfloat[]{ \includegraphics{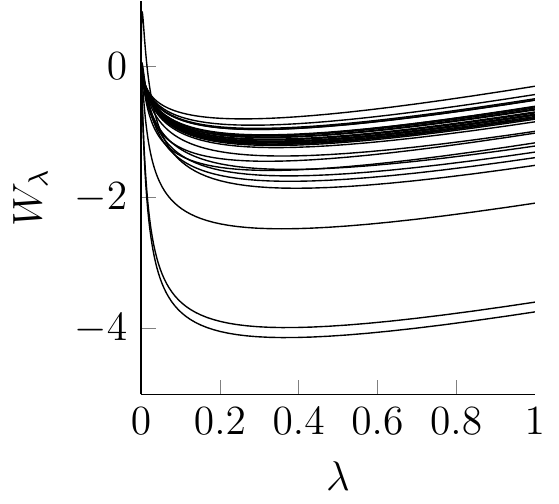} \label{subfig:gaussian_flow_importance_sampling_weight}}
\captionsetup{singlelinecheck=off}
\caption[.]{A simple example of Gaussian flow importance sampling using \eqref{eq:test_model} with $\sigma_{x}=1$ and $\sigma_{y}=0.1$, showing evolution of \subref{subfig:gaussian_flow_importance_sampling_state} the particle states and \subref{subfig:gaussian_flow_importance_sampling_weight} the particle log-weights. Contours of the target posterior are shown with solid lines. Particle paths are shown with dotted lines, with the initial and final state shown by a circle and cross respectively.}
\label{fig:gaussian_flow_importance_sampling}
\end{figure}

In figures~\ref{fig:parameter_plots_sigx} to \ref{fig:parameter_plots_dsf}, we compare the root mean square error (RMSE) and effective sample size (ESS) \citep{Kong1994} obtained using particle flow importance sampling against those obtained from conventional importance sampling using two different choices of importance distribution: the prior, and a Laplace approximation of the posterior formed at the mode. For these results, the particle flow uses a fine fixed grid of pseudo-time steps, in order to give an idea of optimal performance. The dependence on the prior and observation variances and state dimension $\lsdim$ are illustrated, and the effect of the diffusion scale factor $\dsf$.

The Gaussian flow consistently outperforms the simpler samplers, particularly so in the more extreme parameter settings. Particle flow sampling has the greatest advantage when the posterior is ill-conditioned or particularly non-Gaussian. This occurs when the state dimension is greater than the number of observations, and the observations are informative compared to the prior, either because the prior variance is high or the observation variance low. In combination with the nonlinear observation function, this gives rise to complex posterior distributions, in which the modes are often highly ``curved'' (as in figure~\ref{fig:gaussian_flow_importance_sampling}) and thus poorly represented by a single Gaussian.

\begin{figure}[t]
\centering
\subfloat[]{ \includegraphics{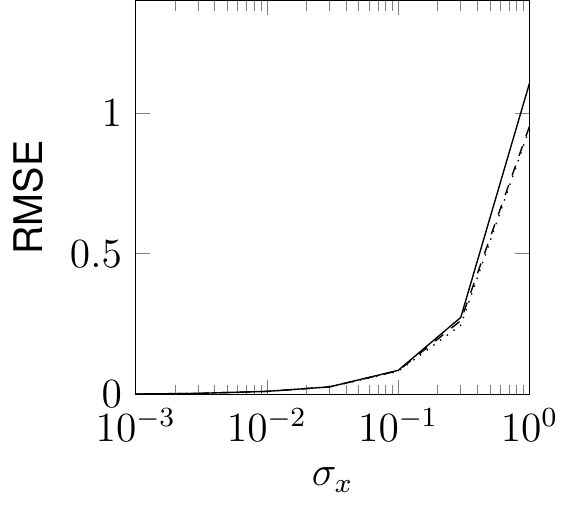} \label{subfig:parameter_testing_rmse_sigx}}
\subfloat[]{ \includegraphics{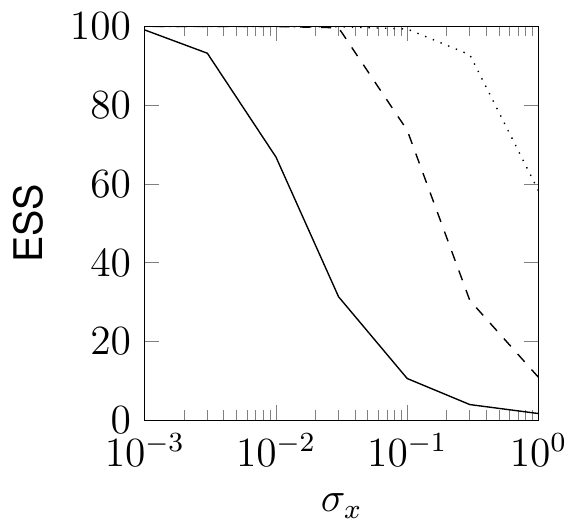} \label{subfig:parameter_testing_ess_sigx}}
\caption[]{Root mean square error \subref{subfig:parameter_testing_rmse_sigx} and effective sample size \subref{subfig:parameter_testing_ess_sigx} with $\lsdim=2$, $\sigma_{y}=0.1$ and varying prior standard deviation $\sigma_{x}$ for three importance samplers. Sampling from the prior (solid), Laplace approximation importance density (dashed), and particle flow with $\dsf=0$ (dotted). $100$ particles for each.}
\label{fig:parameter_plots_sigx}
\end{figure}

\begin{figure}[t]
\centering
\subfloat[]{ \includegraphics{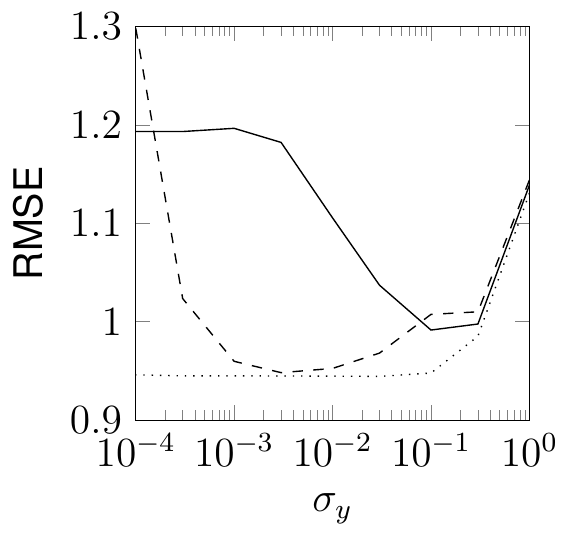} \label{subfig:parameter_testing_rmse_sigy}}
\subfloat[]{ \includegraphics{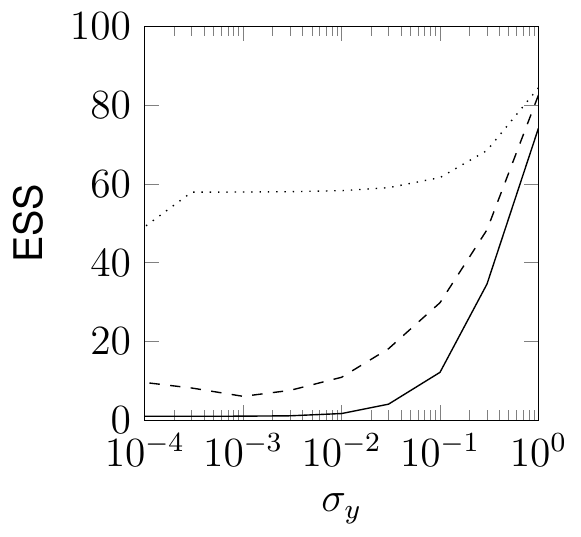} \label{subfig:parameter_testing_ess_sigy}}
\caption[]{Root mean square error \subref{subfig:parameter_testing_rmse_sigy} and effective sample size \subref{subfig:parameter_testing_ess_sigy} with $\lsdim=2$, $\sigma_{x}=1$ and varying observation standard deviation $\sigma_{y}$ for three importance samplers. Sampling from the prior (solid), Laplace approximation importance density (dashed), and particle flow with $\dsf=0$ (dotted). $100$ particles for each.}
\label{fig:parameter_plots_sigy}
\end{figure}

\begin{figure}[t]
\centering
\subfloat[]{ \includegraphics{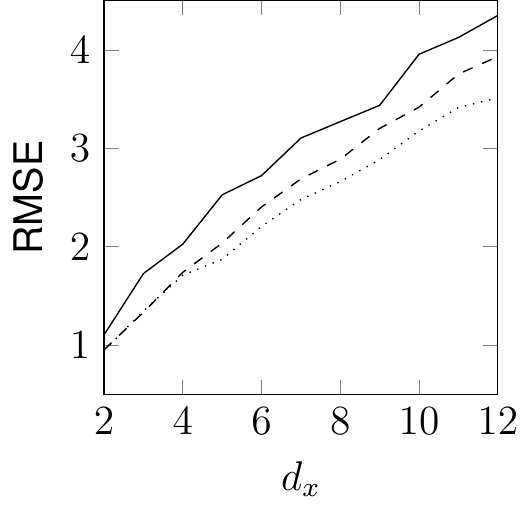} \label{subfig:parameter_testing_rmse_dim}}
\subfloat[]{ \includegraphics{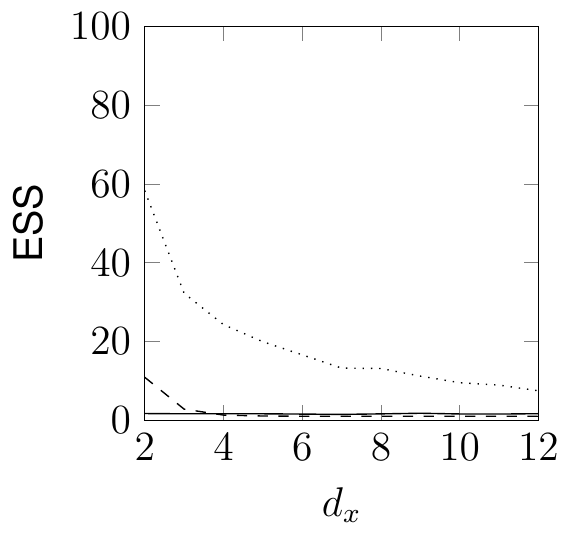} \label{subfig:parameter_testing_ess_dim}}
\caption[]{Root mean square error \subref{subfig:parameter_testing_rmse_dim} and effective sample size \subref{subfig:parameter_testing_ess_dim} with $\sigma_{x}=1$, $\sigma_{y}=0.1$ and varying state dimension for three importance samplers. Sampling from the prior (solid), Laplace approximation importance density (dashed), and particle flow with $\dsf=0$ (dotted). $100$ particles for each.}
\label{fig:parameter_plots_dim}
\end{figure}

\begin{figure}[t]
\centering
\subfloat[]{ \includegraphics{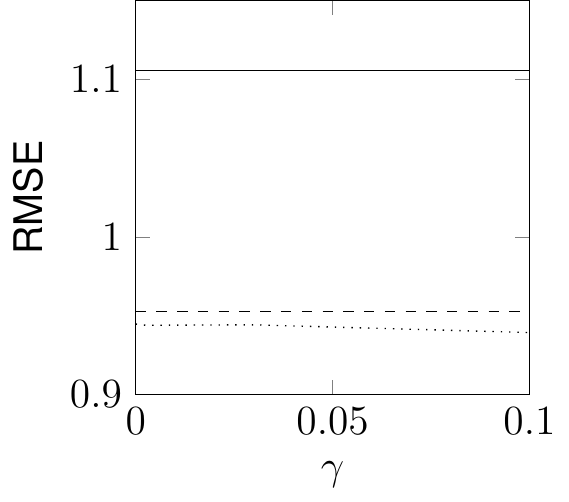} \label{subfig:parameter_testing_rmse_dsf}}
\subfloat[]{ \includegraphics{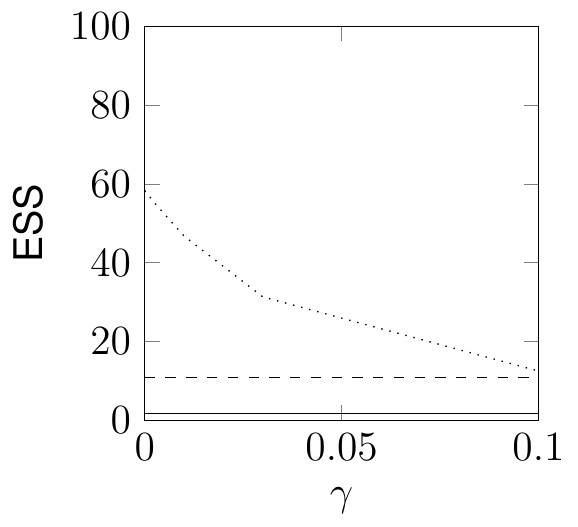} \label{subfig:parameter_testing_ess_dsf}}
\caption[]{Root mean square error \subref{subfig:parameter_testing_rmse_dsf} and effective sample size \subref{subfig:parameter_testing_ess_dsf} with $\lsdim=1$, $\sigma_{x}=1$, $\sigma_{y}=0.1$ for three importance samplers. Sampling from the prior (solid), Laplace approximation importance density (dashed), and particle flow (dotted) with varying diffusion scale factor $\dsf$. $100$ particles for each.}
\label{fig:parameter_plots_dsf}
\end{figure}

In figure~\ref{fig:particle_plots} we provide another demonstration of the benefits of using particle flow importance sampling, this time using a more practical model (a single-frame of the altitude-assisted tracking described in section~\ref{sec:simulations}) and with an implementation of the adaptive step size mechanism. Particle states are shown before and after a resampling step. The number of samples drawn in each case is scaled such that the running time for each is the same. It is clear that the particle flow is better able to characterise the posterior, while doing more than just place particles around a mode.

\begin{figure}
\centering
\subfloat[]{ \includegraphics{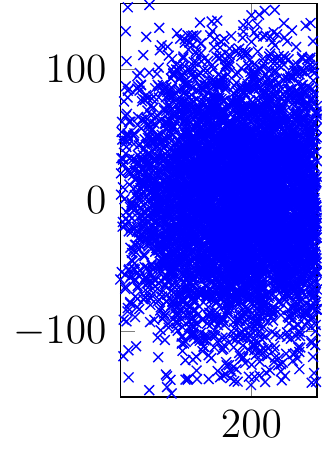}\label{subfig:particles_pre_bootstrap}}
\subfloat[]{ \includegraphics{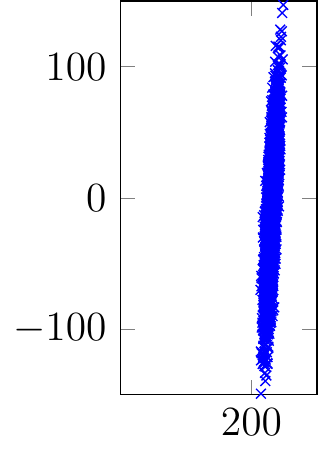}\label{subfig:particles_pre_laplace}}
\subfloat[]{ \includegraphics{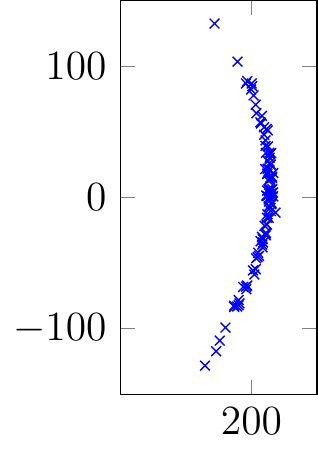}\label{subfig:particles_pre_flow}}\\
\subfloat[]{ \includegraphics{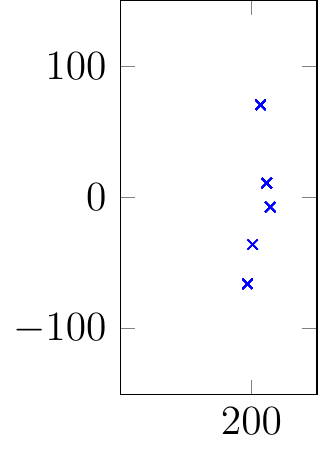}\label{subfig:particles_post_bootstrap}}
\subfloat[]{ \includegraphics{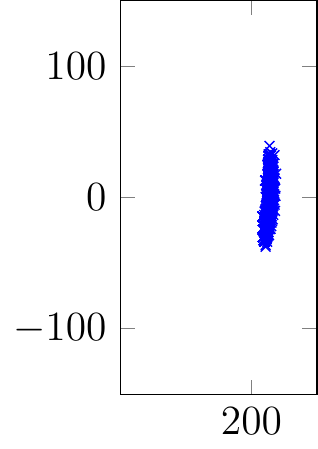}\label{subfig:particles_post_laplace}}
\subfloat[]{ \includegraphics{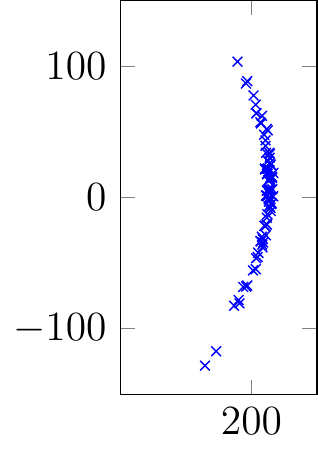}\label{subfig:particles_post_flow}}
\caption[]{Particles sampled before \subref{subfig:particles_pre_bootstrap},\subref{subfig:particles_pre_laplace},\subref{subfig:particles_pre_flow} and after \subref{subfig:particles_post_bootstrap},\subref{subfig:particles_post_laplace},\subref{subfig:particles_post_flow} resampling using three different strategies: sampling from the prior \subref{subfig:particles_pre_bootstrap},\subref{subfig:particles_post_bootstrap}, sampling from a Laplace approximation of the posterior \subref{subfig:particles_pre_laplace},\subref{subfig:particles_post_laplace} and a Gaussian flow \subref{subfig:particles_pre_flow},\subref{subfig:particles_post_flow}.}
\label{fig:particle_plots}
\end{figure}

\subsection{Resample-Move with Particle Flow Proposals} \label{sec:resample_move}

If an importance sampler generates a set of particles which is dominated by a small number with large weights, then the resulting posterior estimates will have a high variance. When this happens, a post-processing stage known as resample-move \citep{Gilks2001} may improve the situation. The weighted particle set is first resampled according to the normalised importance weights to produce an unweighted set. In this standard procedure, low-weight particles are discarded and high-weight particles copied to replace them, with the number of replicates chosen randomly in an appropriate manner so as to ensure unbiasedness \citep{Hol2006,Douc2005}. These replicated particles are then perturbed by sampling from an MCMC kernel so as to spread them around and further explore the promising areas of the state space. Resampling reduces the weight variance of a particle set at the cost of introducing dependence between the particles. The MCMC steps are used to reduce this dependence. Note that the MCMC does not need to be run to convergence in resample-move, since it is being used merely to improve sample diversity.

Implementing resample-move effectively requires some additional algorithm parameters to be selected, such as the number of MCMC steps and an appropriate proposal distribution for Metropolis-Hastings (MH). When particle flow sampling is used, there is an obvious choice for this proposal; simply return to the original state for each particle which was sampled from the prior, $\ls{0}$, and re-simulate a new path through pseudo-time. The choice of proposal distribution is thus reduced to setting a value of $\dsf$, the diffusion scale factor, which will control the size of the proposed moves. Clearly with $\dsf=0$ the motion is deterministic and no move would be taken, i.e. the chain remains stuck in its current location.

Since new values of $\ls{1}$ are to be drawn independently conditional on $\ls{0}$, the acceptance probability is simply that for an MH independence sampler. That is, if the existing state has weight $\pw{}$ (unnormalised, before resampling), and new state for the MH proposal has unnormalised weight $\pw{}\fixed$, then the MH acceptance probability is,
\begin{IEEEeqnarray}{rCl}
 \min\left\{1, \frac{\pw{}\fixed}{\pw{}} \right\}     .
\end{IEEEeqnarray}
Figure~\ref{fig:resample_move} shows two stochastic flows being used for resample-move, illustrating the scope for exploring the state space using this method.

\begin{figure}
\centering
\subfloat[]{\includegraphics{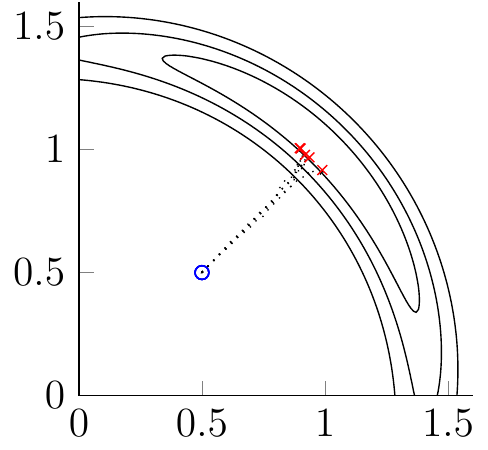}\label{subfig:resample_move_small}}
\subfloat[]{\includegraphics{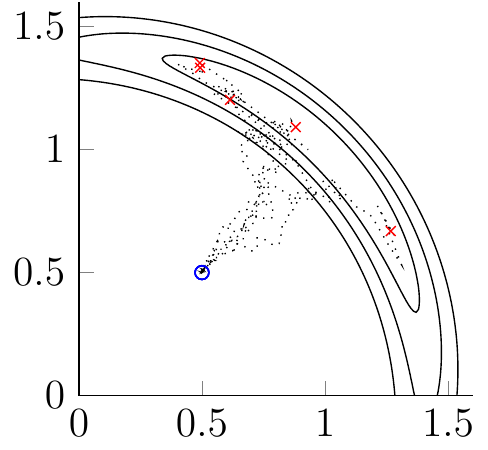}\label{subfig:resample_move_large}}
\caption[]{Particle trajectories for two particle flows targeting the example model \eqref{eq:test_model} using $\sigma_{x}=1$, $\sigma_{y}=0.1$, using \subref{subfig:resample_move_small} $\dsf=0.001$ and \subref{subfig:resample_move_large} $\dsf=0.1$. Contours of the target posterior are shown with solid lines. Particle paths are shown with dotted lines, with the initial and final state shown by a circle and cross respectively. For this easy problem, the acceptance probabilities in both cases are close to $1$.}
\label{fig:resample_move}
\end{figure}

\section{Applications in Particle Filtering} \label{sec:gaussian_flows_for_particle_filters}

Our motivating purpose for studying particle flows is for use in filtering. We consider a standard discrete-time Markovian state space model in which the transition, observation and prior models have closed-form densities,
\begin{IEEEeqnarray}{rClCrClCrCl}
 \ls{\ti} & \sim & \transden(\ls{\ti} | \ls{\ti-1}) & \qquad & \ob{\ti} & \sim & \obsden(\ob{\ti} | \ls{\ti}) & \qquad & \ls{1} & \sim & \priorden(\ls{1})                  \nonumber      ,
\end{IEEEeqnarray}
where the random variable $\ls{\ti}$ is the hidden state of a system at time $\ti$, and $\ob{\ti}$ is an incomplete, noisy observation.

A conventional particle filter \citep{Cappe2007,Doucet2009} uses importance sampling to estimate distributions recursively over the path of the state variables, $\ls{1:\ti}=\{\ls{1}, \dots, \ls{\ti}\}$, such that,
\begin{IEEEeqnarray}{rCl}
 \sum_{i=1}^{\numpart} \npw{\ti}\pss{i} \phi(\ls{1:\ti}\pss{i}) & \rightasconverge & \int \postden(\ls{1:\ti}) \phi(\ls{1:\ti}) d\ls{1:\ti}      \nonumber       .
\end{IEEEeqnarray}
Each step begins by selecting a set of ancestors $\{\anc{\ti}{i}\}$ from amongst the ($\ti-1$)th step particles according to the corresponding weights. Next, a new state is proposed for each particle from an importance density $\ls{\ti}\pss{i} \sim \impden(\ls{\ti} | \ls{\ti-1}\pss{\anc{\ti}{i}}, \ob{\ti})$, and this is concatenated to the ancestral path to form the new particle $\ls{1:\ti}\pss{i} \leftarrow \left\{ \ls{1:\ti-1}\pss{\anc{\ti}{i}},  \ls{\ti}\pss{i} \right\}$. An importance weight is then assigned to the particle to account for the discrepancy between importance and target distributions,
\begin{IEEEeqnarray}{rClCl}
 \pw{\ti}\pss{i} & = & \frac{ \den(\ls{1:\ti}\pss{i} | \ob{1:\ti}) }{ \den(\ls{1:\ti-1}\pss{\anc{\ti}{i}} | \ob{1:\ti-1}) \impden(\ls{\ti}\pss{i} | \ls{\ti-1}\pss{\anc{\ti}{i}}, \ob{\ti}) } & \propto & \frac{ \transden(\ls{\ti}\pss{i} | \ls{\ti-1}\pss{\anc{\ti}{i}}) \obsden(\ob{\ti}|\ls{\ti}\pss{i}) }{ \impden(\ls{\ti}\pss{i} | \ls{\ti-1}\pss{\anc{\ti}{i}}, \ob{\ti}) } \label{eq:particle_filter_weight}     .
\end{IEEEeqnarray}

It was shown by \cite{Doucet2000a} that the incremental weight variance is minimised by proposing from the conditional posterior $\impden(\ls{\ti} | \ls{\ti-1}\pss{\anc{\ti}{i}}, \ob{\ti}) = \den(\ls{\ti} | \ls{\ti-1}\pss{\anc{\ti}{i}}, \ob{\ti})$, known as the optimal importance density (OID). This cannot be used routinely due to an intractable normalising constant required in the weight caluclations.

\subsection{Existing Particle Flow Approaches}

The approach taken by \cite{Daum2008,Daum2011d,Daum2013,Reich2011,Reich2012a} is to apply particle flow sampling directly to the filtering density. Assume that a set of unweighted particles exists approximating $\den(\ls{\ti-1}|\ob{1:\ti-1})$. The predictive density at the next step is related by,
\begin{IEEEeqnarray}{rCl}
 \den(\ls{\ti}|\ob{1:\ti-1}) & = & \int \transden(\ls{\ti}|\ls{\ti-1}) \den(\ls{\ti-1}|\ob{1:\ti-1}) d\ls{\ti-1}     ,
\end{IEEEeqnarray}
which can thus be sampled by simply drawing $\ls{\ti}\pss{i} \sim \transden(\ls{\ti}|\ls{\ti-1}\pss{i})$ for each particle and then marginalising (i.e. discarding) the old states. Defining this predictive density as the prior and the filtering density as the posterior, a particle flow is used to sample from,
\begin{IEEEeqnarray}{rCl}
 \den(\ls{\ti}|\ob{1:\ti}) & = & \frac{\den(\ls{\ti}|\ob{1:\ti-1}) \obsden(\ob{\ti}|\ls{\ti})}{\nconst{\ti}}      .
\end{IEEEeqnarray}
The difficulty with this approach is that finding an appropriate flow generally requires at least the prior and often also its gradient and Hessian to be calculable pointwise. This is not the case for the predictive density, $\den(\ls{\ti}|\ob{1:\ti-1})$. (Note that we could use a Monte Carlo approximation of this density, but the resulting algorithm has a complexity of $\bigo{\numpart^2}$ in the number of particles.) \cite{Reich2011,Reich2012a,Reich2013} address this by making analytical approximations of this density as a Gaussian or Gaussian mixture. \cite{Daum2008,Daum2011d,Daum2013,Daum2009c} use a number of methods, including Gaussian and various numerical approximations. These approximations alter the actual distribution of the particles. The filter no longer returns a properly weighted set of particles representing the posterior and consistent estimates of posterior expectations are no longer guaranteed.

Furthermore, the existing particle flow algorithms do not fall within the framework of ordinary particle filters. They only provide us with an estimate of the marginal filtering density $\den(\ls{\ti}|\ob{1:\ti})$, rather than the more conventional path filtering density $\den(\ls{1:\ti}|\ob{1:\ti})$. This may sometimes be all that is needed, but on other occasions samples of the entire path are essential, for example for smoothing \citep{Kitagawa1996} or parameter estimation schemes, such as particle MCMC \citep{Andrieu2010}.

\subsection{Gaussian Flow Approximations to the Optimal Importance Density}

In this work, we use particle flow sampling within the standard particle filtering framework, thus retaining samples of the entire path and avoiding the need for additional layers of approximation. In order to achieve this, we need to consider two different density sequences. The flow for each particle state is derived by targeting the optimal importance density (OID) with the sequence,
\begin{IEEEeqnarray}{rCl}
 \frac{ \transden(\ls{\ti}|\ls{\ti-1}\pss{\anc{\ti}{i}}) \obsden(\ob{\ti}|\ls{\ti})^{\pt} }{ \nconst{\pt}(\ls{\ti-1}\pss{\anc{\ti}{i}}) }     .
\end{IEEEeqnarray}
This allows us to sample a value for $\ls{\ti}$ conditional on the history $\ls{1:\ti-1}$. Meanwhile, the weight updates are conducted so as to target the filtering density over the entire trajectory, with the sequence,
\begin{IEEEeqnarray}{rCl}
 \frac{ \den(\ls{1:\ti-1}|\ob{1:\ti-1}) \transden(\ls{\ti}|\ls{\ti-1}\pss{\anc{\ti}{i}}) \obsden(\ob{\ti}|\ls{\ti})^{\pt} }{ \nconst{\pt} }     .
\end{IEEEeqnarray}
With this simple modification, the required weight update formula becomes,
\begin{IEEEeqnarray}{rCl}
 \pw{\pt_1} & \propto & \pw{\pt_0} \times \frac{ \obsden(\ob{\ti} | \ls{\pt_1})^{\pt_1} \transden(\ls{\pt_1} | \ls{\ti-1}) }{ \obsden(\ob{\ti} | \ls{\pt_0})^{\pt_0} \transden(\ls{\pt_0} | \ls{\ti-1}) } \times \determ{ \pdv{\ls{\pt_1}}{\ls{\pt_0}} }       .
\end{IEEEeqnarray}

\section{Simulations} \label{sec:simulations}

Numerical testing using simulated data is presented to demonstrate the efficacy of Gaussian flow sampling for particle filtering. We measure performance by considering RMSE values, using the empirical particle mean as a point estimate, and average effective sample size (ESS), measured before resampling \citep{Kong1994}.

The following particle filters (and their respective importance densities) were tested:
\begin{itemize}
        \item A bootstrap filter (BF), using the transition density. \citep{Gordon1993}
        \item An extended particle filter (EPF), using a Gaussian density chosen by linearisation about the predictive mean, in the style of an extended Kalman filter. \citep{Doucet2000a}
        \item An unscented particle filter (UPF), using a Gaussian density chosen using the unscented transform, in the style of an unscented Kalman filter. \citep{Merwe2000}
        \item A Laplace approximation particle filter (LAPF), using a Gaussian density chosen by truncation of the Taylor series of the log of the unnormalised OID around a local maximum \citep{Doucet2000a}. Gradient ascent is used to locate the maximum.
        \item A Gaussian flow particle filter (GFPF), using the the Gaussian flow importance sampling method, with $\dsf=0$. The adaptive step size mechanism is used and requires roughly $5$ to $40$ steps.
\end{itemize}

The posterior filtering distributions of the chosen models can assume complex and irregular shapes, sometimes leading to the complete failure of the EPF and UPF. The LAPF is generally slow because the maximisation procedure struggles with the irregular mode shapes.

The number of particles for the GFPF was set to $100$. For the remaining filters, the number of particles was increased so as to achieve a similar running time. On the altitude-assisted tracking model, the LAPF in fact took roughly 3 times as long as the other algorithms.

\subsection{Models}

\subsubsection{Altitude-Assisted Tracking}
We consider tracking a small aircraft over a mapped landscape, a scenario inspired by \cite{Schon2005}. Time of flight and Doppler measurements from a radio transmitter on the aircraft provide accurate measurements of range $\rng{\ti}$, and range rate $\rngrt{\ti}$, but only a low resolution measurement of bearing $\bng{\ti}$. In addition, accurate measurements are made of the height above the ground $\hei{\ti}$. The profile of the terrain (i.e. the height of the ground above a datum at each point) has been mapped.

At time step $\ti$, the latent state for our model is,
\begin{IEEEeqnarray}{rCl}
 \ls{\ti} & = & \begin{bmatrix} \pos{\ti}^T & \vel{\ti}^T \end{bmatrix}^T \nonumber      ,
\end{IEEEeqnarray}
where $\pos{\ti}$ and $\vel{\ti}$ are the $3$-dimensional position and velocity of the aircraft respectively, and the observation is,
\begin{IEEEeqnarray}{rCl}
 \ob{\ti} & = & \begin{bmatrix} \bng{\ti} & \rng{\ti} & \hei{\ti} & \rngrt{\ti} \end{bmatrix}^T       .
\end{IEEEeqnarray}
The observation function is described by the following equations,
\begin{IEEEeqnarray}{rClCrCl}
 \bng{\ti} & = & \arctan\left(\frac{\pos{\ti,1}}{\pos{\ti,2}}\right) + \noise{\ti,1} & \qquad \qquad & \rng{\ti} & = & \sqrt{ \pos{\ti,1}^2 + \pos{\ti,3}^2 + \pos{\ti,3}^2 }  + \noise{\ti,2} \nonumber \\
 \hei{\ti} & = & \pos{\ti,3} - \terrain( \pos{\ti,1}, \pos{\ti,2} )  + \noise{\ti,3} & \qquad \qquad & \rngrt{\ti} & = & \frac{ \pos{\ti}\cdot\vel{\ti} }{ \rng{\ti} }  + \noise{\ti,4} \nonumber      ,
\end{IEEEeqnarray}
where $\terrain( \pos{\ti,1}, \pos{\ti,2} )$ is the terrain height at the corresponding horizontal coordinates. The four noise terms have independent zero-mean Gaussian densities and the respective variances are $\left(\frac{\pi}{9}\right)^2$, $0.1^2$, $0.1^2$, $0.1^2$. A linear Gaussian near-constant velocity transition model is used \citep{Bar-Shalom1995}, with volatility of $30^2$. The terrain profile was modelled as a mixture of randomly-generated Gaussian blobs. An example is shown in figure~\ref{fig:drone_terrain_map}.

\begin{figure}[bt]
\centering
\includegraphics{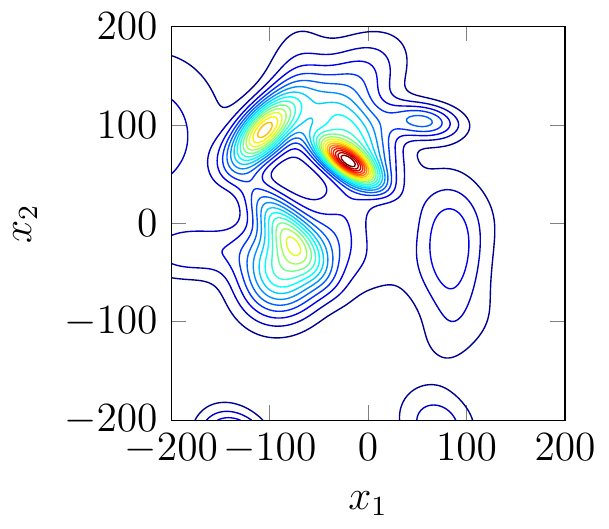}
\caption{Contour plot of an example simulated terrain map.}
\label{fig:drone_terrain_map}
\end{figure}

The accurate measurements of range, range rate and height constrain the region of high posterior probability to lie on a $3$ dimensional subspace, which can take some very irregular shapes.

\subsubsection{Fitting A Skeletal Model}

We consider a toy motion-capture problem, in which camera measurements are used to estimate the pose of a human arm. The latent state consists the 3D coordinates of the shoulder joint $\xS{}$, the orientation $\aB$, the angles of the shoulder $\aS$ and elbow $\aE$, and the lengths of the upper $\dU$ and lower $\dL$ arm. The evolution of each of these is modelled as a random walk with Gaussian noise. The variances are $0.5^2$ for position in the transverse directions, and $0.1^2$ in the depth direction, $\frac{\pi}{18}^2$ for the angles and $0.001^2$ for the lengths (which allows for model inaccuracy, and avoids the need to do static parameter estimation). The observation model consists of two stages. First, the elbow and hand positions are calculated using,
\begin{IEEEeqnarray}{rClCrCl}
 \xE{} = \xS{} + \dU \begin{bmatrix}
                  \cos(\aB)\cos(\aS) \\
                  \sin(\aS) \\
                  \sin(\aB)\cos(\aS)
                 \end{bmatrix} \nonumber & \qquad &
 \xH{} = \xE{} + \dL \begin{bmatrix}
                  \cos(\aB)\cos(\aS+\aE) \\
                  \sin(\aS+\aE) \\
                  \sin(\aB)\cos(\aS+\aE)
                 \end{bmatrix} \nonumber      .
\end{IEEEeqnarray}
Observations of the shoulder and hand positions are made through a perspective projection. By choosing an appropriate coordinate system, this may be modelled simply using,
\begin{IEEEeqnarray}{rCl}
 \obsfun(\ls{}) & = & \begin{bmatrix}
                       \frac{\xS{1}+\xS{3}}{\xS{3}} & \frac{\xS{2}+\xS{3}}{\xS{3}} & \frac{\xE{1}+\xE{3}}{\xE{3}} & \frac{\xE{2}+\xE{3}}{\xE{3}}
                      \end{bmatrix}^T \nonumber      .
\end{IEEEeqnarray}
The observations are accurate, with a variance of $0.001^2$.

\subsection{Results}

Figures~\ref{fig:drone_example_frame_deterministic} and \ref{fig:sineha_example_frame} show the motion of the particles from the GFPF on a typical frame, and the awkward shapes of the posterior mode. Tables~\ref{tab:drone_results} and \ref{tab:sineha_results} show the average ESSs and RMSEs for each algorithm over 100 simulated data sets, each of 100 time steps.

\begin{figure}
\centering
\subfloat[]{\includegraphics{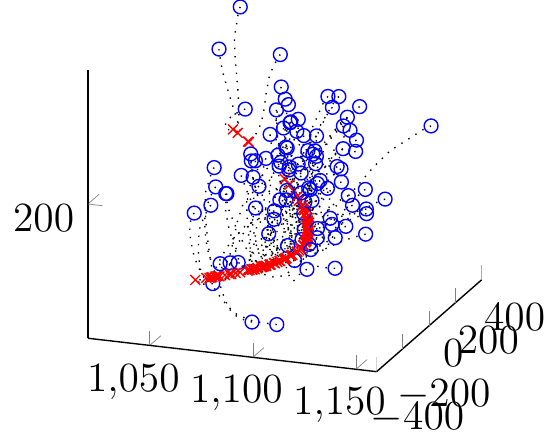}\label{fig:drone_example_frame_deterministic}}
\subfloat[]{\includegraphics{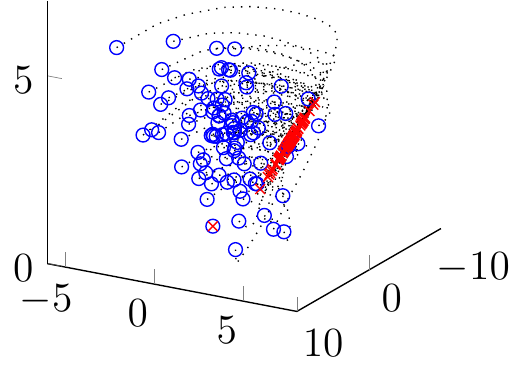}\label{fig:sineha_example_frame}}
\caption{An example of the GFPF particle motion running on the terrain tracking model (a), showing 3D position, and the skeletal arm model (b), showing 3D shoulder position. Prior states are shown with circles and posterior states with crosses.}
\end{figure}

\begin{table}
\centering
\begin{tabular}{l||c|c|c}
Algorithm                                & $N_F$  & ESS  & RMSE \\
\hline
Bootstrap                                & 5000 & 1   & 847 \\
Extended Kalman                          & 1500 & 40  & 417 \\
Unscented Kalman                         & 400  & 18  & 277 \\
Laplace Approximation                    & 100  & 14   & 347 \\
Gaussian Flow                            & 100  & 57  & 171 \\
\end{tabular}
\caption{Algorithm performance results on the altitude-assisted tracking model, showing number of filter particle ($N_F$), effective sample size (ESS), and root mean square error (RMSE).}
\label{tab:drone_results}
\end{table}

\begin{table}
\centering
\begin{tabular}{l||c|c|c}
Algorithm                                & $N_F$ & ESS  & RMSE \\
\hline
Bootstrap                                & 11000 & 1  & 2.6 \\
Extended Kalman                          & 5000  & 17 & 7.2 \\
Laplace Approximation                    & 100   & 5  & 6.8 \\
Gaussian Flow                            & 100   & 58 & 1.3 \\
\end{tabular}
\caption{Algorithm performance results on the skeletal arm model, showing number of filter particle ($N_F$), effective sample size (ESS), and root mean square error (RMSE). The EPF occasionally diverges and fails to complete. These instances are excluded from the results in the table. The UPF regularly fails and is excluded completely.}
\label{tab:sineha_results}
\end{table}

Particle flow resample-move was also tested on the altitude-assisted tracking model. Figure~\ref{fig:drone_example_frame_stochastic} shows the resulting stochastic motion of the particles. Using $\dsf=0.3$, roughly $25$--$50\%$ of the MH steps were accepted. The RMSE performance was not significantly improved.
\begin{figure}
\centering
\subfloat[]{\includegraphics{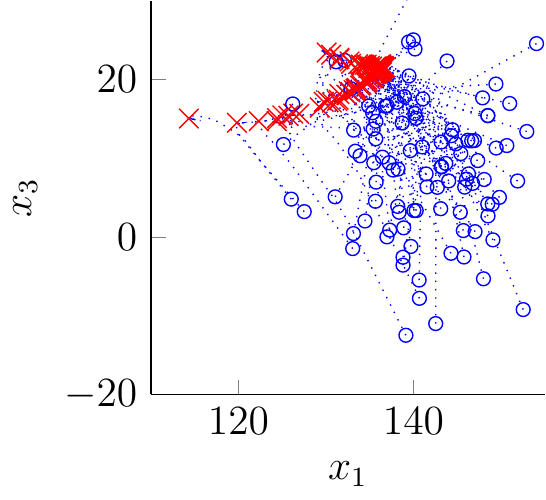}}
\subfloat[]{\includegraphics{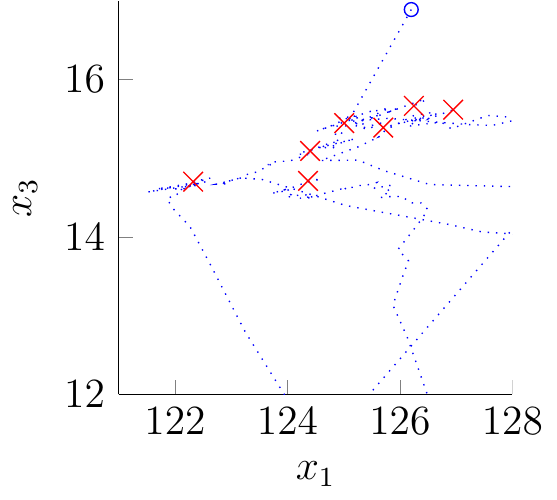}}
\caption{An example of the stochastic GFPF ($\dsf=0.3$) particle motion running on the altitude-assisted tracking model, showing one horizontal and the vertical state component. Prior states are shown with circles and posterior states with crosses. The second panel is a close-up showing the stochastic motion of the particles.}
\label{fig:drone_example_frame_stochastic}
\end{figure}

\section{Discussion and Conclusions}

We have described the use of particle flow importance sampling using an approximate Gaussian flow, and how this may be used to sample approximately from the optimal importance density of a particle filter. The simulations presented in the previous section demonstrate that this procedure is capable of producing better particle approximations (higher effective sample sizes and lower errors) than simpler particle filters (which use a simple Gaussian importance density) on a class of challenging state space models.

The method introduced is appropriate for models with a Gaussian prior and likelihood but highly nonlinear dependence between the observations and latent state. The algorithm requires almost no tuning. The number of particles and the tolerance for the adaptive step-size selection process are the only critical parameters.

The particle flow and optimal transport methods of \citep{Daum2008,Daum2011d,Reich2011,Reich2012a} use similar particle flow ideas to address the task of filtering as we do here. The essential differences in this work are:
\begin{itemize}
  \item We target the optimal importance density rather than the filtering density directly. The OID is known pointwise up to a normalising constant, and thus we avoid the need for one layer of approximation.
  \item In \citep{Daum2008,Daum2011d,Reich2011,Reich2012a}, particle flow samples are used directly to form an approximation of the posterior, with the result that asymptotic convergence properties are lost. We use the particle flow samples as the input to an importance sampler, and correct for the difference between the implied importance density and the posterior density with an appropriate importance weight. \cite{Reich2013} has used a similar importance sampling formulation, but uses different mechanisms to move the particles and assign weights.
  \item We use an improved numerical integration algorithm based on the analytical solution to the optimal Gaussian flow for linear Gaussian models.
\end{itemize}

Particle flow algorithms bear a resemblance to annealing-type strategies \citep{Neal2001,Deutscher2000,Gall2007,DelMoral2006,Godsill2001b,Oudjane2000}, in that both introduce the likelihood progressively. The fundamental difference is that these strategies all use some form of MCMC or resample-move mechanism to remove the weight degeneracy, while particle flow attempts to prevent it happening in the first place. In fact, the two should be seen as complementary. There is no reason why a particle flow could not be used in combination with an annealing scheme. The particles would be moved independently through pseudo-time using a flow, but periodically they are stopped and an intermediate resampling or resample-move step is performed.

Particle flow sampling is only suitable for continuous variables. It should be noted that when the latent state is mixed, with both discrete and continuous components, it is straightforward to sample the discrete component first and then use a particle flow for the continuous part. Furthermore, a number of heavy tailed distributions, including student-t and alpha-stable, can be written as a scale mixture of normals, such that they are Gaussian conditional on an auxiliary scale variable. If this scale variable is sampled first, then a Gaussian flow may be then be used to sample the state. Successful experiments on such models have been conducted already.

In this work we have exclusively used the methods based on the Gaussian flow, due to its stability and desirable analytical solution. Future research will focus on the use of other choices of particle flow for a more general class of models.

\appendix
\singlespacing

\section{Particle Flow Governing Equation: Proof of theorem~\ref{theo:flow_governing_equation}} \label{app:governing_equation}
The proof follows closely the lines taken by \cite{Daum2008}. First, the log-density is,
\begin{IEEEeqnarray}{rCl}
 \logseqden{\pt}(\ls{}) & = & \logprior(\ls{}) + \pt \loglhood(\ls{}) - \log\left(\nconst{\pt}\right) \nonumber     ,
\end{IEEEeqnarray}
where
\begin{IEEEeqnarray}{rClCrCl}
 \logprior(\ls{}) & = & \log\left(\priorden(\ls{})\right) & \qquad \qquad & \loglhood(\ls{}) & = & \log\left(\lhood(\ls{})\right) \nonumber      .
\end{IEEEeqnarray}
Differentiating the log of the normalising constant, we find,
\begin{IEEEeqnarray}{rCl}
 \frac{d}{d\pt}\log\left(\nconst{\pt}\right) & = & \frac{1}{\nconst{\pt}} \frac{d\nconst{\pt}}{d\pt} \nonumber \\
                                               & = & \frac{ \int \priorden(\ls{}) \lhood(\ls{})^\pt \loglhood(\ls{}) d\ls{} }{ \int \priorden(\ls{}) \lhood(\ls{})^\pt d\ls{} } \nonumber \\
                                               & = & \int \seqden{\pt}(\ls{}) \loglhood(\ls{}) d\ls{} = \expect{\seqden{\pt}}\left[ \loglhood \right] \nonumber     ,
\end{IEEEeqnarray}
and so for the log-density,
\begin{IEEEeqnarray}{rCl}
 \pdv{\logseqden{\pt}}{\pt} & = & \loglhood(\ls{}) - \expect{\seqden{\pt}}\left[ \loglhood \right] \label{app-eq:sequence_logdensity}      .
\end{IEEEeqnarray}

Second, the Fokker-Planck equation relates the motion of a particle with the evolution of the density for its position. For a particle at $\ls{\pt}$ moving according to \eqref{eq:state_sde} and with density $\seqden{\pt}$ it states,
\begin{IEEEeqnarray}{rCl}
 \pdv{\seqden{\pt}}{\pt} & = & - \trace\left[ \pdv{}{\ls{\pt}}\left( \flowdrift{\pt}(\ls{\pt}) \seqden{\pt}(\ls{\pt}) \right) \right] + \sum_{ij} \mpdv{}{\ls{\pt,i}}{\ls{\pt,j}}\left( \flowcov{\pt,ij}(\ls{\pt}) \seqden{\pt}(\ls{\pt}) \right) \nonumber \\
 & = & - \seqden{\pt}(\ls{\pt}) \trace\left[ \pdv{\flowdrift{\pt}}{\ls{\pt}} \right] - \pdv{\seqden{\pt}}{\ls{\pt}}^T \flowdrift{\pt}(\ls{\pt}) + \trace\left[ \flowcov{\pt}(\ls{\pt}) \ppdv{\seqden{\pt}}{\ls{\pt}} \right] \nonumber \\
 & & \qquad\qquad +\: 2 \sum_{ij} \pdv{\flowcov{\pt,ij}}{\ls{\pt,i}} \pdv{\seqden{\pt}}{\ls{\pt,j}} + \seqden{\pt}(\ls{\pt}) \sum_{ij} \mpdv{\flowcov{\pt,ij}}{\ls{\pt,i}}{\ls{\pt,j}} \label{app-eq:fokker_planck}     ,
\end{IEEEeqnarray}
where $\flowcov{\pt}(\ls{}) = \half \flowdiffuse{\pt}(\ls{})\flowdiffuse{\pt}(\ls{})^T$. This may be recast to use log-densities using the following identities,
\begin{IEEEeqnarray}{rClCrCl}
 \pdv{\logseqden{\pt}}{\pt} & = & \frac{ 1 }{ \seqden{\pt}(\ls{}) } \pdv{\seqden{\pt}}{\pt} & \qquad & \pdv{\logseqden{\pt}}{\ls{}} & = & \frac{ 1 }{ \seqden{\pt}(\ls{}) } \pdv{\seqden{\pt}}{\ls{}} \nonumber
\end{IEEEeqnarray}
\begin{IEEEeqnarray}{rClCl}
 \npdv{2}{\logseqden{\pt}}{\ls{}} & = & \frac{ \seqden{\pt}(\ls{}) \npdv{2}{\seqden{\pt}}{\ls{}} - \pdv{\seqden{\pt}}{\ls{}} \pdv{\seqden{\pt}}{\ls{}}^T }{ \seqden{\pt}(\ls{})^2 } & = & \frac{ 1 }{ \seqden{\pt}(\ls{}) } \npdv{2}{\seqden{\pt}}{\ls{}} - \pdv{\logseqden{\pt}}{\ls{}}\pdv{\logseqden{\pt}}{\ls{}}^T \nonumber     .
\end{IEEEeqnarray}
Dividing \eqref{app-eq:fokker_planck} through by $\seqden{\pt}(\ls{\pt})$ (assuming that this is nowhere vanishing) we obtain,
\begin{IEEEeqnarray}{rCl}
 \pdv{\logseqden{\pt}}{\pt} & = & -\trace\left[ \pdv{\flowdrift{\pt}}{\ls{\pt}} \right] - \pdv{\logseqden{\pt}}{\ls{\pt}}^T \flowdrift{\pt}(\ls{\pt}) + \trace\left[ \flowcov{\pt}(\ls{\pt}) \npdv{2}{\logseqden{\pt}}{\ls{\pt}} \right] + \pdv{\logseqden{\pt}}{\ls{\pt}}^T \flowcov{\pt}(\ls{\pt}) \pdv{\logseqden{\pt}}{\ls{\pt}} \nonumber \\
 & & \qquad\qquad\qquad +\: 2 \sum_{ij} \pdv{\flowcov{\pt,ij}}{\ls{\pt,i}} \pdv{\logseqden{\pt}}{\ls{\pt,j}} + \sum_{ij} \mpdv{\flowcov{\pt,ij}}{\ls{\pt,i}}{\ls{\pt,j}} \label{app-eq:log_fp}       .
\end{IEEEeqnarray}
Combining the equations for the log-density \eqref{app-eq:sequence_logdensity} with the partial differential equation for the log-density evolution \eqref{app-eq:log_fp}, the governing equation for the optimal particle dynamics is reached. \qed

\section{Evolution of Ideal Importance Weights: Proof of theorem~\ref{theo:ideal_weight}} \label{app:ideal_weight}

Define $\logpartden{\pt}(\ls{\pt})=\log(\partden{\pt}(\ls{\pt}))$, and apply It\={o}'s lemma,
\begin{IEEEeqnarray}{rCl}
 d\logpartden{\pt} & = & \left[ \pdv{\logpartden{\pt}}{\pt} + \pdv{\logpartden{\pt}}{\ls{\pt}}^T\flowdrift{\pt}(\ls{\pt}) + \half \trace\left[ \flowdiffuse{\pt} \flowdiffuse{\pt}^T \npdv{2}{\logpartden{\pt}}{\ls{\pt}} \right] \right] d\pt + \pdv{\logpartden{\pt}}{\ls{\pt}}^T \flowdiffuse{\pt} d\flowbm{\pt}  \nonumber     .
\end{IEEEeqnarray}
Equivalently to \eqref{app-eq:log_fp}, the Fokker-Planck equation tells us that,
\begin{IEEEeqnarray}{rCl}
 \pdv{\logpartden{\pt}}{\pt} & = & -\trace\left[ \pdv{\flowdrift{\pt}}{\ls{\pt}} \right] - \pdv{\logpartden{\pt}}{\ls{\pt}}^T \flowdrift{\pt}(\ls{\pt}) + \trace\left[ \flowcov{\pt}(\ls{\pt}) \npdv{2}{\logpartden{\pt}}{\ls{\pt}} \right] + \pdv{\logpartden{\pt}}{\ls{\pt}}^T \flowcov{\pt}(\ls{\pt}) \pdv{\logpartden{\pt}}{\ls{\pt}} \nonumber \\
 & & \qquad\qquad\qquad +\: 2 \sum_{ij} \pdv{\flowcov{\pt,ij}}{\ls{\pt,i}} \pdv{\logpartden{\pt}}{\ls{\pt,j}} + \sum_{ij} \mpdv{\flowcov{\pt,ij}}{\ls{\pt,i}}{\ls{\pt,j}} \nonumber       .
\end{IEEEeqnarray}
Combining these,
\begin{IEEEeqnarray}{rCl}
 d\logpartden{\pt} & = & \Bigg[ -\trace\left[ \pdv{\flowdrift{\pt}}{\ls{\pt}} \right] + 2\: \trace\left[ \flowcov{\pt}(\ls{\pt}) \npdv{2}{\logpartden{\pt}}{\ls{\pt}} \right] + \pdv{\logpartden{\pt}}{\ls{\pt}}^T \flowcov{\pt}(\ls{\pt}) \pdv{\logpartden{\pt}}{\ls{\pt}} \nonumber \\
 & & \qquad +\: 2 \sum_{ij} \pdv{\flowcov{\pt,ij}}{\ls{\pt,i}} \pdv{\logpartden{\pt}}{\ls{\pt,j}} + \sum_{ij} \mpdv{\flowcov{\pt,ij}}{\ls{\pt,i}}{\ls{\pt,j}} \Bigg] d\pt + \pdv{\logpartden{\pt}}{\ls{\pt}}^T \flowdiffuse{\pt} d\flowbm{\pt}  \nonumber     .
\end{IEEEeqnarray}
Next, using It\={o}'e Lemma for the target sequence log-density, and inserting \eqref{app-eq:sequence_logdensity},
\begin{IEEEeqnarray}{rCl}
 d\logseqden{\pt} & = & \left[ \pdv{\logseqden{\pt}}{\pt} + \pdv{\logseqden{\pt}}{\ls{\pt}}^T\flowdrift{\pt}(\ls{\pt}) + \half \trace\left[ \flowdiffuse{\pt} \flowdiffuse{\pt}^T \npdv{2}{\logseqden{\pt}}{\ls{\pt}} \right] \right] d\pt + \pdv{\logseqden{\pt}}{\ls{\pt}}^T \flowdiffuse{\pt} d\flowbm{\pt}  \nonumber \\
 & = & \left[ \loglhood(\ls{\pt}) - \expectloglhood + \pdv{\logseqden{\pt}}{\ls{\pt}}^T\flowdrift{\pt}(\ls{\pt}) + \trace\left[ \flowcov{\pt} \npdv{2}{\logseqden{\pt}}{\ls{\pt}} \right] \right] d\pt + \pdv{\logseqden{\pt}}{\ls{\pt}}^T \flowdiffuse{\pt} d\flowbm{\pt}  \label{app-eq:logseqden_differential}      .
\end{IEEEeqnarray}
Finally, for the log-weight $\logpw{\pt}=\log(\pw{\pt})$, we have,
\begin{IEEEeqnarray}{rCl}
 d\logpw{\pt} & = & d\logseqden{\pt} - d\logpartden{\pt} \nonumber     ,
\end{IEEEeqnarray}
and substituting the two differentials the results is reached. \qed

\section{Evolution of Practical Importance Weights: Proof of theorem~\ref{theo:practical_weight}}\label{app:practical_weight}

It is well known that any diffusion process may be constructed as the limit of a particular discrete time Markov chain as the step size tends to $0$ \citep{Oksendal2003}. Specifically, for an It\={o} diffusion, if we have time instants at $\pt_{\ti}=\ti\dpt$, then,
\begin{IEEEeqnarray}{rCl}
 \ls{\ti} & = & \ls{\ti-1} + \flowdrift{\ti-1}(\ls{\ti-1}) \dpt + \flowdiffuse{\ti-1}(\ls{\ti-1}) \stdnorm{\ti} \dpt^{\half} \label{app-eq:discretised_sde}       ,
\end{IEEEeqnarray}
where $\{\stdnorm{\ti}\}$ are drawn independently from a standard Gaussian distribution with density $\bmden(\stdnorm{}) = \normalden{\stdnorm{}}{0}{\dpt I}$. We can derive an appropriate differential equation for a particle importance weight by constructing a sequential importance sampler on this system and then taking the limit $\dpt\to0$.

For a particle with density $\partden{\pt}$ to be properly weighted with respect to the target density $\seqden{\pt}$, the ideal importance weight is given in \eqref{eq:ideal_weight}. This is not practical because $\partden{\pt}$ is generally intractable. Instead we construct an extended target distribution over $\{\ls{\ti}, \stdnorm{1}, \stdnorm{2}, \dots, \stdnorm{\ti}\}$, in the manner of a sequential Monte Carlo (SMC) sampler \citep{DelMoral2006},
\begin{IEEEeqnarray}{rCl}
 \seqden{\ti}(\ls{\ti}) \prod_{k=1}^{\ti} \bmden(\stdnorm{k}) \label{app-eq:extended_target}     .
\end{IEEEeqnarray}
Samples are drawn by first simulating $\ls{0}$ from $\priorden$, the state prior, and $\stdnorm{k}$ for $k=1,\dots,\ti$ from $\bmden$, and then recursively applying \eqref{app-eq:discretised_sde}. We can write the inverse of this transformation using Taylor series expansions of $\flowdrift{\ti-1}$ and $\flowdiffuse{\ti-1}$,
\begin{IEEEeqnarray}{rCl}
  \ls{\ti-1} & = & \ls{\ti} - \flowdrift{\ti}(\ls{\ti}) \dpt - \flowdiffuse{\ti}(\ls{\ti}) \stdnorm{\ti} \dpt^{\half} + \pdv{\left[ \flowdiffuse{\ti}(\ls{\ti}) \stdnorm{\ti} \right]}{\ls{\ti}} \flowdiffuse{\ti}(\ls{\ti}) \stdnorm{\ti} \dpt + \bigo{\dpt^{\frac{3}{2}}} \nonumber     .
\end{IEEEeqnarray}
By the change of variables formula, and using the Jacobian of this inverse transformation, the proposal density in the extended space is,
\begin{IEEEeqnarray}{rCl}
 \priorden(\ls{0}) \prod_{k=1}^{\ti} \bmden(\stdnorm{k}) \determ{\pdv{\ls{k-1}}{\ls{k}}} \label{app-eq:extended_proposal}     .
\end{IEEEeqnarray}
The resulting importance weight is the ratio of target \eqref{app-eq:extended_target} and proposal \eqref{app-eq:extended_proposal} densities in the extended space. Taking the log,
\begin{IEEEeqnarray}{rCl}
 \logpw{\ti} & = & \logseqden{\ti}(\ls{\ti}) - \logprior(\ls{0}) - \sum_{k=1}^{\ti} \log\left(\determ{\pdv{\ls{k-1}}{\ls{k}}}\right) \nonumber \\
 & = & \logpw{\ti-1} + \logseqden{\ti}(\ls{\ti}) - \logseqden{\ti-1}(\ls{\ti-1}) - \log\left(\determ{\pdv{\ls{\ti-1}}{\ls{\ti}}}\right) \label{app-eq:discrete_log_weight}      .
\end{IEEEeqnarray}
To calculate the Jacobian term, we use the following identities,
\begin{IEEEeqnarray}{rCl}
 \determ{I+\delta A} & = & 1 + \delta \: \trace\left[A\right] + \half \delta^2 \left( \trace\left[A\right]^2 - \trace\left[A^2\right] \right) + \bigo{\delta^3} \nonumber \\
 \log(1 + \delta a) & = & \delta a - \frac{\delta^2}{2} a^2  + \bigo{\delta^3} \nonumber      ,
\end{IEEEeqnarray}
with which we reach,
\begin{IEEEeqnarray}{rCl}
 -\log\left(\determ{\pdv{\ls{\ti-1}}{\ls{\ti}}}\right) & = & \trace\left[ \pdv{\flowdrift{\ti-1}}{\ls{\ti-1}} \right] \dpt + \sum_{ij} \pdv{\flowdiffuse{\ti-1,ij}}{\ls{\ti-1,i}} \dbm{\ti,j} - \half \sum_{ijk} \left[ \pdv{\flowdiffuse{\ti-1,ik}}{\ls{\ti-1,j}} \pdv{\flowdiffuse{\ti-1,jk}}{\ls{\ti-1,i}} \right] \dpt + \bigo{\dpt^{\frac{3}{2}}} \nonumber      .
\end{IEEEeqnarray}
Finally, letting $\dpt\to0$ and substituting \eqref{app-eq:logseqden_differential} into \eqref{app-eq:discrete_log_weight}, we obtain the result. \qed

\section{Integrated Gaussian Flow: Proof of Theorem~\ref{theo:integrated_gaussian_flow}} \label{app:integrated_gaussian_flow}

For a small increment of pseudo-time, such that $\pt_0 = \pt$ and $\pt_1 = \pt+\dpt$,
\begin{IEEEeqnarray}{rCl}
 \ls{\pt+\dpt} & = & \lsmn{\pt+\dpt} + \exp\left\{-\half\dsf\dpt\right\} \left(\lsvr{\pt+\dpt}\lsvr{\pt}^{-1}\right)^{\half}(\ls{\pt}-\lsmn{\pt}) + \left[ \frac{1-\exp\left\{-\dsf\dpt\right\}}{\dpt} \right]^{\half} \lsvr{\pt+\dpt}^{\half} \left(\flowbm{\pt_1}-\flowbm{\pt_0}\right) \nonumber     .
\end{IEEEeqnarray}
Now use the following expansions for small increments,
\begin{IEEEeqnarray}{rCl}
 \lsmn{\pt+\dpt} &=& \lsmn{\pt} + \pdv{\lsmn{\pt}}{\pt}\dpt + \bigo{\dpt^2} \nonumber \\
 \lsvr{\pt+\dpt}^{\half} &=& \lsvr{\pt}^{\half} + \pdv{\lsvr{\pt}^{\half}}{\pt}\dpt + \bigo{\dpt^2} \nonumber \\
 \lsvr{\pt+\dpt} &=& \left(\lsvr{\pt}^{-1} + \dpt\lgmom^T\lgmov^{-1}\lgmom + \bigo{\dpt^2}\right)^{-1} \nonumber \\
                 &=& \lsvr{\pt} - \dpt\lsvr{\pt}\lgmom^T\lgmov^{-1}\lgmom\lsvr{\pt} + \bigo{\dpt^2} \nonumber \\
 \left(\lsvr{\pt+\dpt}\lsvr{\pt}^{-1}\right)^{\half} &=& \left( I - \dpt\lsvr{\pt}\lgmom^T\lgmov^{-1}\lgmom + \bigo{\dpt^2} \right)^{\half} \nonumber \\
                                      &=& I - \half\dpt\lsvr{\pt}\lgmom^T\lgmov^{-1}\lgmom + \bigo{\dpt^2} \nonumber \\
 \exp\left\{-\half\dsf\dpt\right\} &=& 1-\half\dsf\dpt + \bigo{\dpt^2}\nonumber \\
 \left[ \frac{1-\exp\left\{-\dsf\dpt\right\}}{\dpt} \right]^{\half} &=& \left[ \frac{\dsf\dpt + \bigo{\dpt^2}}{\dpt} \right]^{\half} = \dsf^{\half} + \bigo{\dpt} \nonumber      ,
\end{IEEEeqnarray}
and noting that,
\begin{IEEEeqnarray}{rCl}
 \pdv{\lsmn{\pt}}{\pt} & = & \lsvr{\pt} \lgmom^T\lgmov^{-1}\left(\ob{}-\lgmom\lsmn{\pt}\right) \nonumber      ,
\end{IEEEeqnarray}
leads to,
\begin{IEEEeqnarray}{rCl}
 \ls{\pt+\dpt} - \ls{\pt} & = & \lsvr{\pt}\lgmom^T\lgmov^{-1} \left[ \left(\ob{}-\lgmom\lsmn{\pt}\right) - \half\lgmom(\ls{\pt}-\lsmn{\pt}) - \half\dsf(\ls{\pt}-\lsmn{\pt}) \right]\dpt + \dsf^{\half}\lsvr{\pt}^{\half} \left(\flowbm{\pt+\dpt}-\flowbm{\pt}\right) + \bigo{\dpt^{\frac{3}{2}}} \nonumber     .
\end{IEEEeqnarray}
Taking the limit as $\dpt\to0$, the result follows. \qed

\section{Weight Numerical Integration: Proof of Theorem~\ref{theo:weight_numerical_integration}} \label{app:weight_numerical_integration}

Using Taylor expansions,
\begin{IEEEeqnarray}{rCl}
 \log\left(\determ{\pdv{\ls{\pt_1}}{\ls{\pt_0}}}\right) & = & - \log\left(\determ{\pdv{\ls{\pt_0}}{\ls{\pt_1}}}\right) + \bigo{(\pt_1-\pt_0)^{\frac{3}{2}}} \nonumber     .
\end{IEEEeqnarray}
Hence, the log-weight update for a small step is described by \eqref{app-eq:discrete_log_weight} with an error of order $\bigo{(\pt_1-\pt_0)^{\frac{3}{2}}}$. Therefore, by construction, as the step size goes to $0$, the weight evolves according to \eqref{eq:practical_weight_differential_equation}, and from theorem~\ref{theo:practical_weight}, the particle is properly weighted. \qed

\singlespacing
\bibliographystyle{apalike}
\bibliography{/users/pete/Dropbox/PhD/Cleanbib}

\begin{thebibliography}{}

\bibitem[Andrieu et~al., 2010]{Andrieu2010}
Andrieu, C., Doucet, A., and Holenstein, R. (2010).
\newblock Particle {M}arkov chain {M}onte {C}arlo methods.
\newblock {\em Journal of the Royal Statistical Society: Series B (Statistical
  Methodology)}, 72:269--342.

\bibitem[Bar-Shalom and Li, 1995]{Bar-Shalom1995}
Bar-Shalom, Y. and Li, X.~R. (1995).
\newblock {\em Multitarget-multisensor tracking: principles and techniques}.
\newblock Storrs, CT : Yaakov Bar-Shalom.

\bibitem[Bartels and Stewart, 1972]{Bartels1972}
Bartels, R.~H. and Stewart, G.~W. (1972).
\newblock Solution of the matrix equation ax + xb = c.
\newblock {\em Commun. ACM}, 15(9):820--826.

\bibitem[Bunch and Godsill, 2013]{Bunch2013a}
Bunch, P. and Godsill, S. (2013).
\newblock Particle filtering with progressive {G}aussian approximations to the
  optimal importance density.
\newblock In {\em 5th IEEE International Workshop on Computational Advances in
  Multi-Sensor Adaptive Processing (CAMSAP)}.

\bibitem[Capp{\'e} et~al., 2007]{Cappe2007}
Capp{\'e}, O., Godsill, S., and Moulines, E. (2007).
\newblock An overview of existing methods and recent advances in sequential
  {Monte Carlo}.
\newblock {\em Proceedings of the IEEE}, 95:899--924.

\bibitem[Daum and Huang, 2008]{Daum2008}
Daum, F. and Huang, J. (2008).
\newblock Particle flow for nonlinear filters with log-homotopy.
\newblock In {\em Proceedings of SPIE, the International Society for Optical
  Engineering}. Society of Photo-Optical Instrumentation Engineers.

\bibitem[Daum and Huang, 2011]{Daum2011d}
Daum, F. and Huang, J. (2011).
\newblock Particle degeneracy: root cause and solution.
\newblock In {\em Proc. SPIE}, volume 8050. SPIE.

\bibitem[Daum and Huang, 2013]{Daum2013}
Daum, F. and Huang, J. (2013).
\newblock Particle flow with non-zero diffusion for nonlinear filters.
\newblock In {\em Proceedings of SPIE 8745, Signal Processing, Sensor Fusion,
  and Target Recognition XXII}, volume 8745. SPIE.

\bibitem[Daum et~al., 2009]{Daum2009c}
Daum, F., Huang, J., Krichman, M., and Kohen, T. (2009).
\newblock Seventeen dubious methods to approximate the gradient for nonlinear
  filters with particle flow.
\newblock In Drummond, O.~E. and Teichgraeber, R.~D., editors, {\em Proc.
  SPIE}, volume 7445. SPIE.

\bibitem[Del~Moral et~al., 2006]{DelMoral2006}
Del~Moral, P., Doucet, A., and Jasra, A. (2006).
\newblock Sequential {M}onte {C}arlo samplers.
\newblock {\em Journal of the Royal Statistical Society: Series B (Statistical
  Methodology)}, 68(3):411--436.

\bibitem[Deutscher et~al., 2000]{Deutscher2000}
Deutscher, J., Blake, A., and Reid, I. (2000).
\newblock Articulated body motion capture by annealed particle filtering.
\newblock In {\em IEEE Conference on Computer Vision and Pattern Recognition},
  volume~2, pages 126--133.

\bibitem[Douc et~al., 2005]{Douc2005}
Douc, R., Cappe, O., and Moulines, E. (2005).
\newblock Comparison of resampling schemes for particle filtering.
\newblock In {\em Proc. 4th Int. Symp. Image and Signal Processing and
  Analysis}.

\bibitem[Doucet et~al., 2000]{Doucet2000a}
Doucet, A., Godsill, S., and Andrieu, C. (2000).
\newblock On sequential {M}onte {C}arlo sampling methods for {B}ayesian
  filtering.
\newblock {\em Statistics and Computing}, 10:197--208.

\bibitem[Doucet and Johansen, 2009]{Doucet2009}
Doucet, A. and Johansen, A.~M. (2009).
\newblock A tutorial on particle filtering and smoothing: Fifteen years later.
\newblock In Crisan, D. and Rozovsky, B., editors, {\em The Oxford Handbook of
  Nonlinear Filtering}. Oxford University Press.

\bibitem[Gall et~al., 2007]{Gall2007}
Gall, J., Potthoff, J., Schn\"{o}rr, C., Rosenhahn, B., and Seidel, H.-P.
  (2007).
\newblock Interacting and annealing particle filters: Mathematics and a recipe
  for applications.
\newblock {\em Journal of Mathematical Imaging and Vision}, 28(1):1--18.

\bibitem[Gelman and Meng, 1998]{Gelman1998}
Gelman, A. and Meng, X.-L. (1998).
\newblock Simulating normalizing constants: From importance sampling to bridge
  sampling to path sampling.
\newblock {\em Statistical Science}, 13(2):163--185.

\bibitem[Geweke, 1989]{Geweke1989}
Geweke, J. (1989).
\newblock {B}ayesian inference in econometric models using {M}onte {C}arlo
  integration.
\newblock {\em Econometrica}, 57(6):1317--1339.

\bibitem[Gilks and Berzuini, 2001]{Gilks2001}
Gilks, W.~R. and Berzuini, C. (2001).
\newblock Following a moving target --- {M}onte {C}arlo inference for dynamic
  {B}ayesian models.
\newblock {\em Journal of the Royal Statistical Society: Series B (Statistical
  Methodology)}, 63(1):127--146.

\bibitem[Godsill and Clapp, 2001]{Godsill2001b}
Godsill, S. and Clapp, T. (2001).
\newblock Improvement strategies for {M}onte {C}arlo particle filters.
\newblock In Doucet, A., de~Freitas, N., and Gordon, N., editors, {\em
  Sequential Monte Carlo Methods in Practice}, pages 139--158. Springer New
  York.

\bibitem[Gordon et~al., 1993]{Gordon1993}
Gordon, N.~J., Salmond, D.~J., and Smith, A. F.~M. (1993).
\newblock Novel approach to nonlinear/non-{G}aussian {B}ayesian state
  estimation.
\newblock {\em IEE Proceedings F, Radar and Signal Processing},
  140(2):107--113.

\bibitem[Hagmar et~al., 2011]{Hagmar2011}
Hagmar, J., Jirstrand, M., Svensson, L., and Morelande, M. (2011).
\newblock Optimal parameterization of posterior densities using homotopy.
\newblock In {\em 14th International Conference on Information Fusion
  (FUSION)}.

\bibitem[Hanebeck and Steinbring, 2012]{Hanebeck2012}
Hanebeck, U. and Steinbring, J. (2012).
\newblock Progressive {G}aussian filtering.
\newblock {\em arXiv preprint arXiv:1204.0133}.

\bibitem[Hanebeck and Feiermann, 2003]{Hanebeck2003a}
Hanebeck, U.~D. and Feiermann, O. (2003).
\newblock Progressive {B}ayesian estimation for nonlinear discrete-time
  systems:the filter step for scalar measurements and multidimensional states.
\newblock In {\em 42nd IEEE Conference on Decision and Control}, volume~5,
  pages 5366--5371.

\bibitem[Hol et~al., 2006]{Hol2006}
Hol, J.~D., Schon, T.~B., and Gustafsson, F. (2006).
\newblock On resampling algorithms for particle filters.
\newblock In {\em IEEE Nonlinear Statistical Signal Processing Workshop}, pages
  79--82.

\bibitem[Kitagawa, 1996]{Kitagawa1996}
Kitagawa, G. (1996).
\newblock {M}onte {C}arlo filter and smoother for non-{G}aussian nonlinear
  state space models.
\newblock {\em Journal of Computational and Graphical Statistics}, 5(1):1--25.

\bibitem[Kong et~al., 1994]{Kong1994}
Kong, A., Liu, J.~S., and Wong, W.~H. (1994).
\newblock Sequential imputations and bayesian missing data problems.
\newblock {\em Journal of the American statistical association},
  89(425):278--288.

\bibitem[Liu, 2001]{Liu2001a}
Liu, J.~S. (2001).
\newblock {\em {M}onte {C}arlo strategies in scientific computing}.
\newblock Springer.

\bibitem[Neal, 2001]{Neal2001}
Neal, R.~M. (2001).
\newblock Annealed importance sampling.
\newblock {\em Statistics and Computing}, 11(2):125--139.

\bibitem[{\O}ksendal, 2003]{Oksendal2003}
{\O}ksendal, B. (2003).
\newblock {\em Stochastic differential equations: {An} introduction with
  applications}.
\newblock Springer Verlag.

\bibitem[Oudjane and Musso, 2000]{Oudjane2000}
Oudjane, N. and Musso, C. (2000).
\newblock Progressive correction for regularized particle filters.
\newblock In {\em 3rd International Conference on Information Fusion (FUSION)},
  volume~2, pages 10--17.

\bibitem[Reich, 2011]{Reich2011}
Reich, S. (2011).
\newblock A dynamical systems framework for intermittent data assimilation.
\newblock {\em BIT Numerical Mathematics}, 51:235--249.

\bibitem[Reich, 2012]{Reich2012a}
Reich, S. (2012).
\newblock A {G}aussian-mixture ensemble transform filter.
\newblock {\em Quarterly Journal of the Royal Meteorological Society},
  138(662):222--233.

\bibitem[Reich, 2013]{Reich2013}
Reich, S. (2013).
\newblock A guided sequential monte carlo method for the assimilation of data
  into stochastic dynamical systems.
\newblock In {\em Recent Trends in Dynamical Systems}, pages 205--220.
  Springer.

\bibitem[Sch\"{o}n et~al., 2005]{Schon2005}
Sch\"{o}n, T., Gustafsson, F., and Nordlund, P.-J. (2005).
\newblock Marginalized particle filters for mixed linear/nonlinear state-space
  models.
\newblock {\em IEEE Transactions on Signal Processing}, 53(7):2279--2289.

\bibitem[Shampine and Reichelt, 1997]{Shampine1997}
Shampine, L. and Reichelt, M. (1997).
\newblock The {MATLAB} {ODE} suite.
\newblock {\em SIAM Journal on Scientific Computing}, 18(1):1--22.

\bibitem[Van Der~Merwe et~al., 2000]{Merwe2000}
Van Der~Merwe, R., Doucet, A., Freitas, N., and Wan, E. (2000).
\newblock The unscented particle filter.
\newblock {\em Advances in Neural Information Processing Systems}, 13.

\end{thebibliography}

\end{document}